\documentclass[a4paper]{amsproc}
\usepackage{amssymb}
\usepackage{amscd}
\usepackage{amsfonts}
\usepackage{amssymb}
\usepackage{cite}
\usepackage{amsmath}
\usepackage{epsfig}
\usepackage{graphicx, subfigure}
\usepackage{color, graphics}
\usepackage{pdfsync}
\usepackage{hyperref}
\usepackage[all,cmtip]{xy}
\usepackage{setspace}

\theoremstyle{theorem}
\newtheorem{theorem}{Problem}
\newtheorem{prop}[theorem]{Proposition}

\theoremstyle{definition}
\newtheorem*{dfn}{Definition}
\newtheorem{rem}{Remark}

\newcommand{\R}{\mathbb{R}}

\allowdisplaybreaks

\begin{document}
\title[Affine Geometry and Relativity]{Affine Geometry and Relativity}
\author{Bo\v zidar Jovanovi\'c}

\address{Mathematical Institute SANU \\
Serbian Academy of Sciences and Arts \\
Kneza Mihaila 36, 11000 Belgrade\\
Serbia}
\email{bozaj@mi.sanu.ac.rs}

\subjclass{51N10, 51N30, 51P05, 70A05, 83A05}

\keywords{Affine transformations, the Galilean principle of relativity, the Galilean and
pseudo-Euclidean geometry, addition of velocities, the Iwasawa decomposition}

\maketitle

\begin{abstract}
We present the basic concepts of space and time, the Galilean and
pseudo-Euclidean geometry. We use an elementary
geometric framework of affine spaces and groups of affine
transformations to illustrate the natural relationship between
classical mechanics and theory of relativity, which is quite often  hidden,
despite its fundamental importance.
We have emphasized a passage
from the group of Galilean motions to the group of Poincar\'e
transformations of a plane. In particular, a
1-parametric family of natural deformations of the Poincar\'e
group is described. We also
visualized the underlying groups of Galilean, Euclidean, and
pseudo-Euclidean rotations within the special linear group.
\end{abstract}


\section{Introduction.}

We present concepts of space and time, the Galilean and pseudo-Euclidean geometry at an elementary level.
In Section \ref{relativnost} we recall on the Galilean principle
of relativity in the affine world (space-time) of events. This
fundamental principle uses the notion of inertial frames that need
additional geometrical structure. Depending on the structure, we
have two different mechanics: the classical one and the special
relativity, which share the same principle of relativity.
In Sections \ref{galilejevSvet} and \ref{galilejeveTransformacije}
we define the Galilean world and  derive the group of Galilean motion:
the group of affine transformations that relates coordinates
between different inertial frames. We work mainly with a
two-dimensional world, which is sufficient to realize the main
ingredients and differences between the Euclidean, Galilean, and
pseudo-Euclidean geometries.

In Section \ref{poenkareovSvet} we adopt some of Poincar\'e ideas to emphasize a passage from the group of
Galilean motions $SG(2)$ to the group of Poincar\'e transformations of a plane $SP^+(1,1)$, while in Section \ref{pseudoSvet} we define the basic notion of the
underlying pseudo-Euclidean geometry.
 The full group of affine transformations $G(2)$ and $P(1,1)$ that preserve the Galilean
and pseudo-Euclidean structures are described in Section \ref{pune}.

Finally, in Section \ref{vizualizacija},
we visualize the special
linear group $SL(2)$  and its subgroups: special orthogonal group
$SO(2)$, group of Galilean rotations $SG_0(2)$, and the special
orthogonal group $SO(1,1)$ of signature $(1,1)$.
By using the Iwasawa decomposition,
we realize  $SL(2)$ as a space $\R^3$ without a line.
We also visualize $SL(2)$ as a quadric in $\R^4$ by considering several intersections with 3-dimensional affine subspaces of $\R^4$.

We employed the very basics of the theory of groups, affine and
vector spaces,  affine transformations and Euclidean
isometries of the plane. The calculus is applied only in Remarks
\ref{PrvaPrimedba} and \ref{DrugaPrimedba}.
Also, at the end of Section \ref{poenkareovSvet},
we used  Lie groups to describe a 1-parametric family of natural deformations $SP^+_{c,a}(1,1)$ of the group of Poincar\'e
transformations of the two-dimensional affine world.
Note that the group $SP^+_{c,a}(1,1)$ leads to the same law of addition of velocities as in special relativity, showing that one should be very careful
in derivation of the Lorentz and Poincar\'e transformations.  Curiously, we did not find the group $SP^+_{c,a}(1,1)$ in the literature.

This paper can serve as a supplementary
instructional material in teaching linear algebra, geometry, and
classical mechanics at various levels, from high school (a part of the paper) to undergraduate
and graduate studies. To that end, several
statements are listed in the form of problems, which are left as exercises
for readers and students. This material can be used, for example,
in initial lectures in graduate courses on analytical mechanics and
symplectic geometry \cite{Ar, LM}, and the theory of relativity \cite{HE}. It
illustrates the appearance of Lie groups in
mechanics and provides a necessary geometric background for the
formulation of the Newtonian laws of motion. Our goal has also been to
underline the natural and unifying relationship between classical mechanics
and the theory of relativity, which is quite often forgotten or hidden,
despite its fundamental importance.

\section{Affine world and the principle of relativity.}\label{relativnost}

One of the most profound steps in the understanding of nature, known today as \emph{the principle of relativity},
was discovered by Galileo Galilei (1564--1642) in his
\emph{Dialogo sopra i due massimi sistemi del mondo} from 1632. \footnote{The book is written in the form of discussions. As an illustration, we present the following beautiful part, taken from \cite{Y}, page 18:
"When you have observed all these things carefully (though there is
no doubt that when the ship is standing still everything must happen in this
way), have the ship proceed with any speed you like, so long as the motion is
uniform and not fluctuating this way and that. You will discover not the least
change in all the effects named, nor could you tell from any of them whether
the ship was moving or standing still. In jumping, you will pass on the floor
the same spaces as before, nor will you make larger jumps toward the stem
than toward the prow even though the ship is moving quite rapidly, despite
the fact that during the time that you are in the air the floor under you will
be going in a direction opposite to your jump. In throwing something to your
companion, you will need no more force to get it to him whether he is in the
direction of the bow or the stem, with yourself situated opposite. The
droplets will fall as before into the vessel beneath without dropping toward
the stem, although while the drops are in the air the ship runs many spans."}
At the end of the XIX century, the Galilean principle of relativity was formulated in modern terms and popularized by Henri Poincar\'e (1854--1912), see \cite{P1}. It can be stated as follows (see Vladimir Arnold (1937--2010) \cite{Ar}).

\begin{itemize}

\item[(R1)] \emph{All the laws of nature at all moments of time are the same in all inertial coordinate systems.}

\item[(R2)] \emph{A coordinate system in uniform rectilinear motion with respect to an inertial one is also inertial.}

\end{itemize}

Here we assumed that our space is the 3-dimensional Euclidean space and that the time is 1-dimensional, i.e., described by one real parameter.
In total, we have that all events can be seen as points  in a 4-dimensional affine space $\mathcal A^4$, called the \emph{world}
(or the \emph{universe}, or the \emph{space-time}).
In other words, all events in an  \emph{inertial coordinate system} are described by 4 coordinates $(x,y,z,t)$. The coordinates $(x,y,z)$ are Cartesian  coordinates in the 3-dimensional Euclidean space and $t$ is a time coordinate, measured from a particular event.
Also, from the experience in nature, we consider the Euclidean space with orientation. For example, we need the cross product in the formulation of
many physical laws.

Uniform motions are described by lines, which is consistent with the fact that affine transformations
between different coordinates systems map lines to lines.

What is missing here is the definition of an inertial coordinate system. For
example, in the 3-dimensional Euclidean space, the natural class
of coordinate systems are Cartesian  coordinate systems $(x,y,z)$
determined by frames $[O,\mathbf e_1\mathbf e_2\mathbf
e_3]$, where $\mathbf e_1,\mathbf e_2,\mathbf e_3$ is an
orthonormal basis of the associated vector space.
Thus, we need to define an additional structure in the affine world $\mathcal A^4$ that will catch both the Euclidean structure and the time.
We will see that depending on the structure, we obtain two different mechanics: the classical one and the special relativity.
However, in both cases we apply the same Galilean principle of relativity R1, R2.

\section{Galilean geometry.}\label{galilejevSvet}

The points of the 4-dimensional affine world $\mathcal A^4$ are called \emph{events}.
Let  $\mathbb V$ be the associated 4-dimensional vector space.
The \emph{Galilean structure} $(\mathcal A^4,\tau,\rho)$ is defined  by  (we follow \cite{Ar}, see also \cite{Artz, Le}):

\begin{itemize}
\item[(a)] Non-trivial affine \emph{time} mapping $\tau\colon\mathcal A^4 \to \mathbb R$.
The \emph{oriented time between events} $A$ and $B$ is $l(A,B)=\tau(B)-\tau(A)$ (see Fig.\ref{galilejev svet}).
If $l(A,B)>0$ we say that the event $A$ happened before the event $B$.
The \emph{time length} of a vector $\mathbf w\in\mathbb V$ is $l(\mathbf w)=l(A,B)$, where $\mathbf w=\overrightarrow{AB}$.
It is a non-trivial linear mapping $l\colon \mathbb V\to \mathbb R$.

\item[(b)] 3-dimensional affine subspaces $\mathcal A^3_a=\tau^{-1}(a)$,
called  the \emph{spaces of simultaneous events}, are endowed with the Euclidean structure
\[
\rho(A_1,A_2)=\vert \overrightarrow{A_1A_2}\vert= \sqrt{\langle \overrightarrow{A_1A_2},\overrightarrow{A_1A_2}\rangle}, \qquad l(A_1,A_2)=0,
\]
where $\langle \cdot,\cdot\rangle$ is the scalar product on $l^{-1}(0)\subset \mathbf V$.
We also fix the orientation on the spaces of simultaneous events such that the translations
\begin{equation}\label{translacije}
T_{\overrightarrow{AB}}\colon \mathcal A^3_a\to \mathcal A^3_b, \quad T_{\overrightarrow{AB}}(X)=X+\overrightarrow{AB}, \quad \tau(A)=a, \quad \tau(B)=b
\end{equation}
are isometries that preserve the orientation.
\end{itemize}

Vice verse, a non-trivial time length linear mapping $l\colon \mathbb V\to \mathbb R$ defines the oriented time mapping
$l:\mathcal A^4\times \mathcal A^4\to\R$ and the affine time mapping
$\tau_O\colon\mathcal A^4 \to \mathbb R$,
$\tau_O(A)=l(O,A)$,
$O\in\mathcal A^4$.
That is why we also consider $(\mathcal A^4,l,\rho)$ as a \emph{Galilean structure}.

The vectors with time length equal to zero are called \emph{space-like} vectors. The vectors with non-zero time length are called \emph{time-like} vectors.
Similarly we define \emph{space-like} and \emph{time-like} lines.

\begin{figure}[ht]
{\centering
{\includegraphics[width=8.5cm]{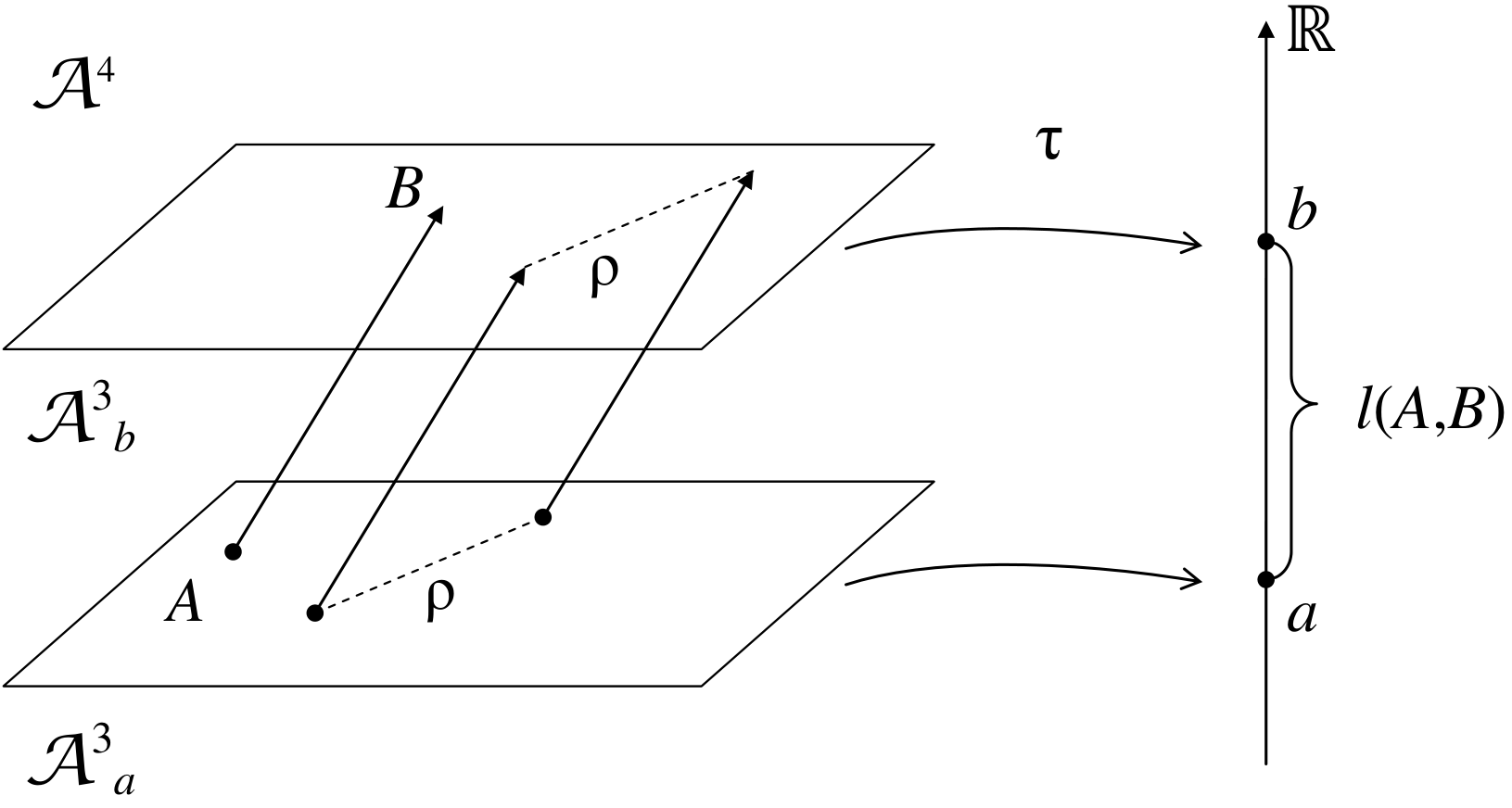}}
\caption{Galilean world} \label{galilejev svet}}
\end{figure}

A motion in the affine world $\mathcal A^4$ along a time-like line $\gamma $ ("forward" in time)  is called a \emph{uniform motion}.
We define the \emph{time-like velocity} of $\gamma$ by:
\begin{equation}\label{TLV}
\mathbf V=\frac{\overrightarrow{AB}}{l(A,B)}, \quad A,B\in\gamma, \quad A\ne B.
\end{equation}
In other words, $\mathbf V$ is the normalized  direction of $\gamma$ ($l(\mathbf V)=1$), see Fig. \ref{promena koordinata}.

Now we are ready to define an inertial coordinate system.

\begin{dfn}
A \emph{Galilean basis} $\mathbf e_1,\mathbf e_2,\mathbf
e_3,\mathbf e_4$ of the vector space $\mathbb V$ is the basis such
that $\mathbf e_1,\mathbf e_2,\mathbf e_3$ are space-like vectors
that form a positive oriented orhonormal basis for the associated
spaces of simultaneous events and $\mathbf e_4$ is a time-like
vector of time length equals to 1. An event $O\in \mathcal A^4$
and a Galilean basis $\mathbf e_1,\mathbf e_2,\mathbf e_3,\mathbf e_4$
define an \emph{inertial reference frame}
$\mathbf I=[O,\mathbf e_1\mathbf e_2\mathbf e_3\mathbf e_4]$ in the Galilean world $\mathcal A^4$.
\end{dfn}

The coordinates $(x,y,z,t)\in \mathbb R^4$ (called \emph{inertial coordinates})
of an event $X\in\mathcal A^4$ with respect to the inertial frame $\mathbf I$ are defined, as usual,  by the identity
\[
\overrightarrow{OX}=x\mathbf e_1+y\mathbf e_2+\mathbf e_3+t\mathbf e_4.
\]
In such a way, we obtain a bijection $\varphi_\mathbf I: \mathcal A^4\to\mathbb R^4$, $X\mapsto \mathbf x=(x,y,z,t)=\varphi_\mathbf I(X)$.
Also, as usual, by $X(x,y,z,t)$ we denote an event $X\in\mathcal A^4$, which has the coordinates $\mathbf x=(x,y,z,t)$
in the inertial reference frame $\mathbf I$.

In an inertial coordinate system $(x,y,z,t)$,
all uniform motions $\gamma$ can be seen as the graphs of the mappings $t\mapsto (x_0+v_x t,y_0+v_y t, z_0+v_z t)$, $t\in\mathbb R$:
\[
\gamma=\{\gamma(t)=(x_0+v_x t,y_0+v_y t, z_0+v_z t,t)\,\vert\, t\in\R\},
\]
where
\[
v_x=\frac{\Delta x}{\Delta t}=\frac{x_2-x_1}{t_2-t_1}, \quad v_y=\frac{\Delta y}{\Delta t}=\frac{y_2-y_1}{t_2-t_1}, \quad
v_z=\frac{\Delta z}{\Delta t}=\frac{z_2-z_1}{t_2-t_1},
\]
and $A_1(x_1,y_1,z_1,t_1)\in\gamma$ and  $A_2(x_2,y_2,z_2,t_2)\in\gamma$ are arbitrary different events.

The space-like vector $\mathbf v=v_x \mathbf e_1+v_y\mathbf e_2+v_z\mathbf e_3$ is called the \emph{velocity} of the uniform motion $\gamma$ in the
inertial coordinate system $(x,y,z,t)$ (see Fig. \ref{promena koordinata}).
The uniform motions in the direction of $\mathbf e_4$ ($\mathbf v=0$) are \emph{at rest} in the coordinate system  $(x,y,z,t)$.
It is clear that for any uniform motion, there exists an inertial reference frame such that this motion is at rest.

Note that the time-like velocity \eqref{TLV} and the space-like velocity $\mathbf v$ of $\gamma$ are related by:
\begin{equation}\label{vremenskaBrzina}
\mathbf V=\frac{\overrightarrow{A_1A_2}}{l(A_1,A_2)}=v_x\mathbf e_1+v_y\mathbf e_2+v_z\mathbf e_3+\mathbf e_4=\mathbf v+\mathbf e_4.
\end{equation}

\begin{rem}\label{PrvaPrimedba}
In general, the \emph{motions in the affine world} $\mathcal A^4$ (or the \emph{world lines}) are curves of the form
\begin{equation}\label{kretanje}
\gamma=\{\gamma(t)=(x(t),y(t), z(t),t)\,\vert\, t\in\R\},
\end{equation}
where $x(t)$, $y(t)$, $z(t)$ are smooth functions of the time $t$ in the chosen inertial reference frame $\mathbf I$. The 
\emph{velocity} and the \emph{time-like velocity} of the motion
\eqref{kretanje} at the moment $t$ are defined, respectively, by 
\[
\mathbf v(t)=\frac{dx}{dt} \mathbf e_1+\frac{dy}{dt}\mathbf e_2+\frac{dz}{dt}\mathbf e_3, \qquad
\mathbf V(t)=\frac{d\gamma}{dt}=\mathbf v(t)+\mathbf e_4.
\]

Note that $\gamma$ is a smooth curve in the affine world $\mathcal A^4$, while the projection
of $\gamma$ to the Euclidean space $\{(x(t),y(t), z(t))\,\vert\, t\in\R\}$ not need to be smooth, even locally.
For example, let us consider the curve $\gamma$ with $x(t)=t^2$, $y(t)=t^3$, $z(t)=0$.
It is a smooth curve in $\mathcal A^4$ (the tangent line at $\gamma(t)$
is spanned by $\mathbf V(t)=2t\mathbf e_1+3t^2\mathbf e_2+\mathbf e_4$).
The projection of $\gamma$ has the cusp singularity at $t=0$.
\end{rem}

\section{Galilean transformations.}\label{galilejeveTransformacije}

Let us consider two inertial reference frames $\mathbf I=[O,\mathbf e_1\mathbf e_2\mathbf e_3\mathbf e_4]$ and
$\mathbf I'=[O',\mathbf e_1'\mathbf e_2'\mathbf e_3'\mathbf e_4']$.
Consider a point $X\in\mathcal A^4$ that has coordinates
\[
\mathbf x=(x,y,z,t)=\varphi_\mathbf I(X), \qquad \mathbf x'=(x',y',z',t')=\varphi_{\mathbf I'}(X)
\]
in the frames $\mathbf I$ and $\mathbf I'$, respectively.
The affine mapping
$G\colon \mathbb R^4\to \mathbb R^4$
defined by
\[
G=\varphi_\mathbf I\circ \varphi_{\mathbf I'}^{-1}, \qquad \mathbf x=A\mathbf x'+\mathbf b,
\] (see Fig. \ref{promena koordinata})
is called the \emph{Galilean transformation} between inertial coordinates of the frames $\mathbf I'$ and $\mathbf I$.
Here $\mathbf b=(b_1,b_2,b_3,b_4)$ are the coordinates of the event $O'$ in the frame $\mathbf I$.

\begin{figure}[ht]
{\centering
{\includegraphics[width=8.5cm]{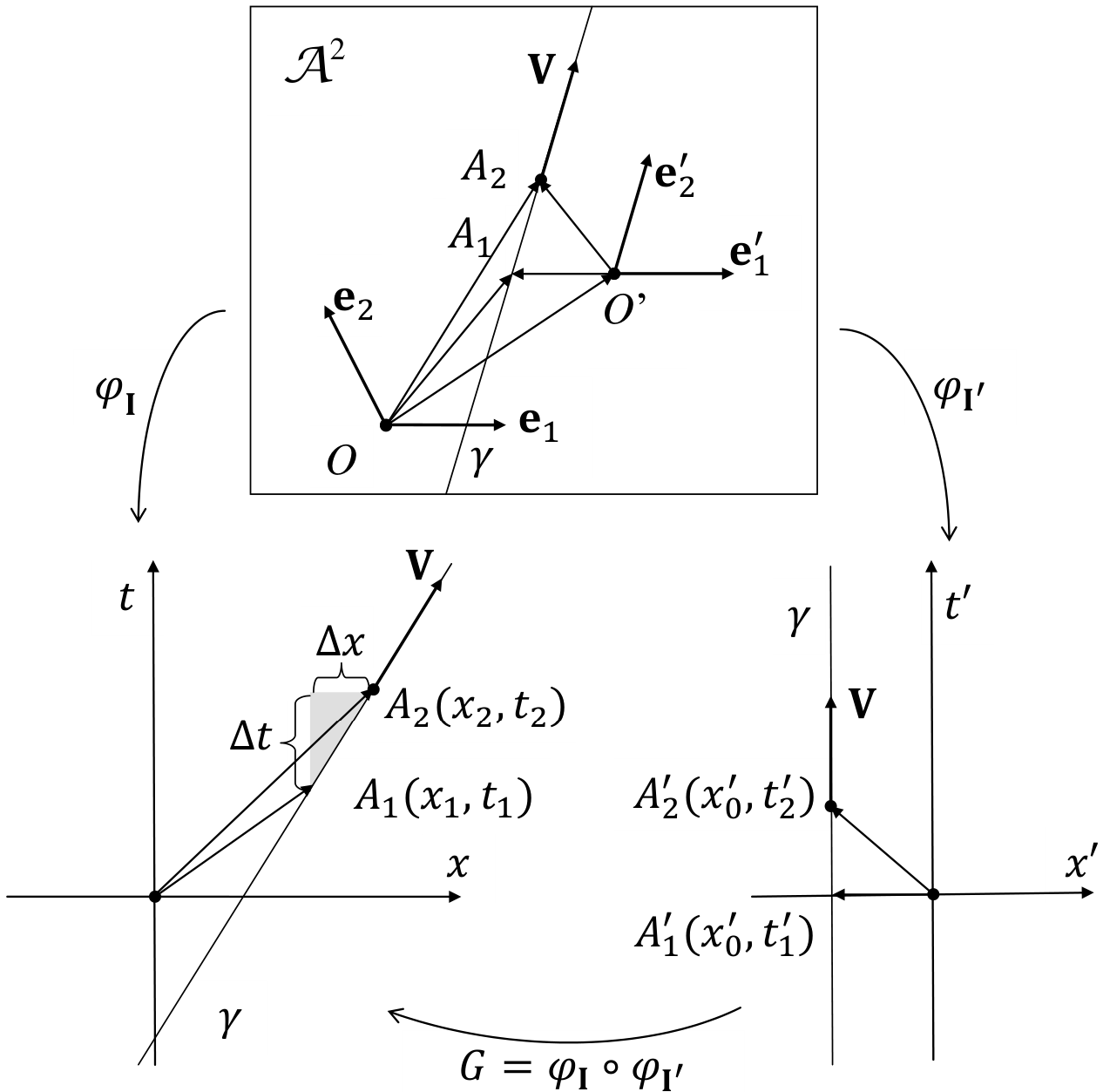}}
\caption{Galilean transformation of coordinates. The uniform motion $\gamma$ that is at rest in the reference frame $\mathbf I'$ has the velocity $\mathbf u=u\mathbf e_1$,
$u=\Delta x/\Delta t$, in the reference frame $\mathbf I$. Note that $u$ equals to the tangent of the angle between $\gamma$ and $t$-axis in the reference frame $\mathbf I$. Also, the time-like velocity $\mathbf V$ of $\gamma$ is equal to $\mathbf e_2'$.} \label{promena koordinata}}
\end{figure}

Instead  to describe the above class of $(4\times 4)$-matrixes $A$, we will
pass to the 2-dimensional affine space-time $\mathcal A^2$ with one-dimensional Euclidean spaces of simultaneous events.
This is sufficient to realize the main differences and similarities between the Euclidean and the Galilean geometry.

With the same definitions as above
we need to find the Galilean transformation
\[
G\colon \mathbb R^2\to \mathbb R^2, \qquad \mathbf x=\mathrm A\mathbf x'+\mathbf b,
\]
where $\mathbf x=(x,t)$ and $\mathbf x'=(x',t')$ are coordinates of an event $X\in \mathcal A^2$ in the inertial frames
$\mathbf I=[O,\mathbf e_1\mathbf e_2]$ and $\mathbf I'=[O',\mathbf e_1'\mathbf e_2']$ (see Fig. \ref{promena koordinata}). Here:
\begin{equation}\label{G-uslovi}
l(\mathbf e_1)=l(\mathbf e_1')=0, \quad \vert \mathbf e_1\vert=\vert \mathbf e_1'\vert=1, \quad l(\mathbf e_2)=l(\mathbf e_2')=1.
\end{equation}

\begin{prop}\label{stav1}
The Galilean transformation
$G$ between the coordinates with respect to the inertial frames $\mathbf I'$ and $\mathbf I$ is of the form
\begin{equation}\label{G-u}
\begin{pmatrix}
x\\
t
\end{pmatrix}=
\begin{pmatrix}
1 & u\\
0 & 1
\end{pmatrix}
\begin{pmatrix}
x'\\
t'
\end{pmatrix}+
\begin{pmatrix}
b_1\\
b_2
\end{pmatrix},
\end{equation}
where $\mathbf b=(b_1,b_2)$ are the coordinates of $O'$ in the frame $\mathbf I$.
The parameter $u$ has the following kinematical interpretation: the motions that are at rest in the reference frame $\mathbf I'$ have the velocity
$\mathbf u=u\mathbf e_1$ in the reference frame $\mathbf I$.
\end{prop}

\begin{proof}
Since $\mathbf e_1$ and $\mathbf e_1'$ are of the same orientation, we have $\mathbf e_1=\mathbf e_1'$.\footnote{
In the 4-dimensional affine world $\mathcal A^4$, two orthonormal bases are related by the elements of the group $SO(3)$.} On the other hand, from \eqref{G-uslovi}
and the fact that the time length $l$ is a linear mapping, we get
\begin{equation}\label{LGT}
\mathbf e_2'=u\mathbf e_1+\mathbf e_2,
\end{equation}
for some real parameter $u$.
Then, for a given event $X$, we have
\begin{align*}
\overrightarrow{OX}=\overrightarrow{OO'}+\overrightarrow{O'X}
&\Longleftrightarrow x\mathbf e_1+t\mathbf e_2=(b_1\mathbf e_1+b_2\mathbf e_2)+(x'\mathbf e_1'+t'\mathbf e_2') \\
&\Longleftrightarrow x\mathbf e_1+t\mathbf e_2=b_1\mathbf e_1+b_2\mathbf e_2+x'\mathbf e_1+t'(u\mathbf e_1+\mathbf e_2)
\end{align*}

By comparing the terms with $\mathbf e_1$ and $\mathbf e_2$ we get the Galilean transformation $G$:
\begin{equation}
(x,t)=(x'+ut'+b_1,t'+b_2),
\end{equation}
which proves \eqref{G-u}.

Next, consider a uniform motion that is at rest in the reference frame $\mathbf I'$:
\[
\gamma=\{(x_0',t')\,\vert\,t'\in\mathbb R\}.
\]

Let $A_1$ and $A_2$ be two arbitrary different events on $\gamma$ with coordinates
$A_i(x_i,t_i)$, $A_i(x_0',t_i')$ $(i=1,2)$
in the reference frames $\mathbf I$ and $\mathbf I'$, respectively  (see Fig. \ref{promena koordinata}).
Then, from \eqref{G-u} we get:
\[
\frac{\Delta x}{\Delta t}=\frac{x_2-x_1}{t_2-t_1}=\frac{(x_0'+ut_2'+b_1)-(x_0'+ut_1'+b_1)}{(t_2'+b_2)-(t_1'+b_2)}=\frac{u(t_2'-t_1')}{t_2'-t_1'}=u.
\]
Therefore, the motion $\gamma$ has the velocity
$\mathbf u=u\mathbf e_1$ in the reference frame $\mathbf I$.

There is a direct proof of the last statement. Namely, the time-like velocity $\mathbf V$ of $\gamma$ is $\mathbf e_2'$ (see Fig. \ref{promena koordinata}).
From \eqref{LGT} we obtain $\mathbf V=u\mathbf e_1+\mathbf e_2$, which  proves that $\mathbf u=u\mathbf e_1$
is the velocity of $\gamma$ (equation \eqref{vremenskaBrzina} considered in the 2-dimensional world).
\end{proof}

Let us denote the Galilean transformation \eqref{G-u} by $G(u,\mathbf b)=G(u,b_1,b_2)$.
In the classical mechanics it is usual to say that the inertial frame $\mathbf I'$ moves with the \emph{velocity} $\mathbf u=u\mathbf e_1$
with respect to the inertial frame $\mathbf I$. 

It is clear that if we have 3 inertial reference frames $\mathbf I$, $\mathbf I'$ and $\mathbf I''$, and if $G(u,\mathbf b)$,
$G(u',\mathbf b')$, and $G(u'',\mathbf b'')$ are coordinate transformations between coordinates in the frames $\mathbf I'$ and $\mathbf I$, $\mathbf I''$ and $\mathbf I'$,
and $\mathbf I''$ and $\mathbf I$, then
$
G(u'',\mathbf b'')=G(u,\mathbf b)\circ G(u',\mathbf b').
$
Whence, the set of the Galilean transformations
\[
SG(2)=\{G(u,b_1,b_2)\,\vert\, u,b_1,b_2\in\mathbb R\}
\]
form a subgroup in the group of affine transformations
of the plane $\mathbb R^2$ with the neutral given by the identity mapping $\mathbf E=G(0,0,0)$.

By the analogy with the Euclidean geometry,
we define the \emph{Galilean rotations} as the
subgroup $SG_0(2)<SG(2)$ of the linear transformations with matrixes of the form
\[
\mathrm A(u)=\begin{pmatrix}
1 & u\\
0 & 1
\end{pmatrix}, \qquad u\in\mathbb R.
\]

\begin{theorem}\label{prvi}
Prove that the group multiplications and the inverses in $SG_0(2)$ and $SG(2)$ are given by:
\begin{align}
\label{sg1} & \mathrm A(u_1)\circ \mathrm A(u_2)=\mathrm A(u_1+u_2), \\
\label{sg2}& \mathrm A^{-1}(u)=\mathrm A(-u),\\
\label{sg3}& G(u,\mathbf b)\circ G(u',\mathbf b')=G(u+u',\mathrm A(u)\mathbf b'+\mathbf b),\\
\label{sg4}&G(u,\mathbf b)^{-1}=G(-u,-\mathrm A(-u)\mathbf b)=G(-u,-b_1+ub_2,-b_2)
\end{align}
\end{theorem}

\begin{theorem} [Addition of velocities in the classical mechanics]
Assume that a time-like line $\gamma$ represent a uniform motion with velocities $\mathbf v=v\mathbf e_1$ and $\mathbf v'=v\mathbf e_1$ in the
inertial frames $\mathbf I$ and $\mathbf I'$, respectively.
If the velocity of the frame $\mathbf I'$ with respect to $\mathbf I$ is $\mathbf u=u\mathbf e_1$, then
\[
v=u+v'.
\]
\begin{itemize}
\item[(i)]  Prove the statement directly, by using the definition $v={\Delta x}/{\Delta t}$ and the Galilean transformation \eqref{G-u}.
\item[(ii)] Prove the statement by using Problem \ref{prvi}
\end{itemize}
\end{theorem}

Note that in \eqref{G-u} we have the transformation between the different coordinates of the same point in the affine plane $\mathcal A^2$.
Besides, we can consider the mappings of the affine plane itself.
More precisely, let us fix the inertial frame $\mathbf I$, i.e., the bijection $\varphi_\mathbf I$ that identifies $\mathcal A^2$
and $\mathbb R^2$. Then \eqref{G-u} defines the affine transformation of $\mathcal A^2$.
Let $\mathbf x_1=(x_1,t_1)$ and $\mathbf x_2=(x_2,t_2)$ be the coordinates of the events $A_1$ and $A_2$ and let
\[
\mathbf x_1'=G(u,\mathbf b)(\mathbf x_1)=A(u)\mathbf x_1+\mathbf b, \qquad \mathbf x_2'=G(u,\mathbf b)(\mathbf x_2)=A(u)\mathbf x_2+\mathbf b,
\]
 be their images $A_1'$ and $A_2'$ in the same coordinates $(x,t)$ (see Fig. \ref{galilejevo kretanje}).  Then
\begin{align*}
& l(A_1,A_2)=t_2-t_1=0  \, \Rightarrow \, l(A_1',A_2')=t_2'-t_1'=0, \,\rho(A_1',A_2')=\rho(A_1,A_2) \\
& l(A_1,A_2)=t_2-t_1 \neq 0 \, \Rightarrow \, l(A_1',A_2')=t_2'-t_1'=t_2-t_1.
\end{align*}

Thus, the mapping $G(u,\mathbf b)$ is a \emph{Galilean isometry} -- it preserves the
Galilean structure of $\mathcal A^2$.

\begin{figure}[ht]
{\centering
{\includegraphics[width=8.5cm]{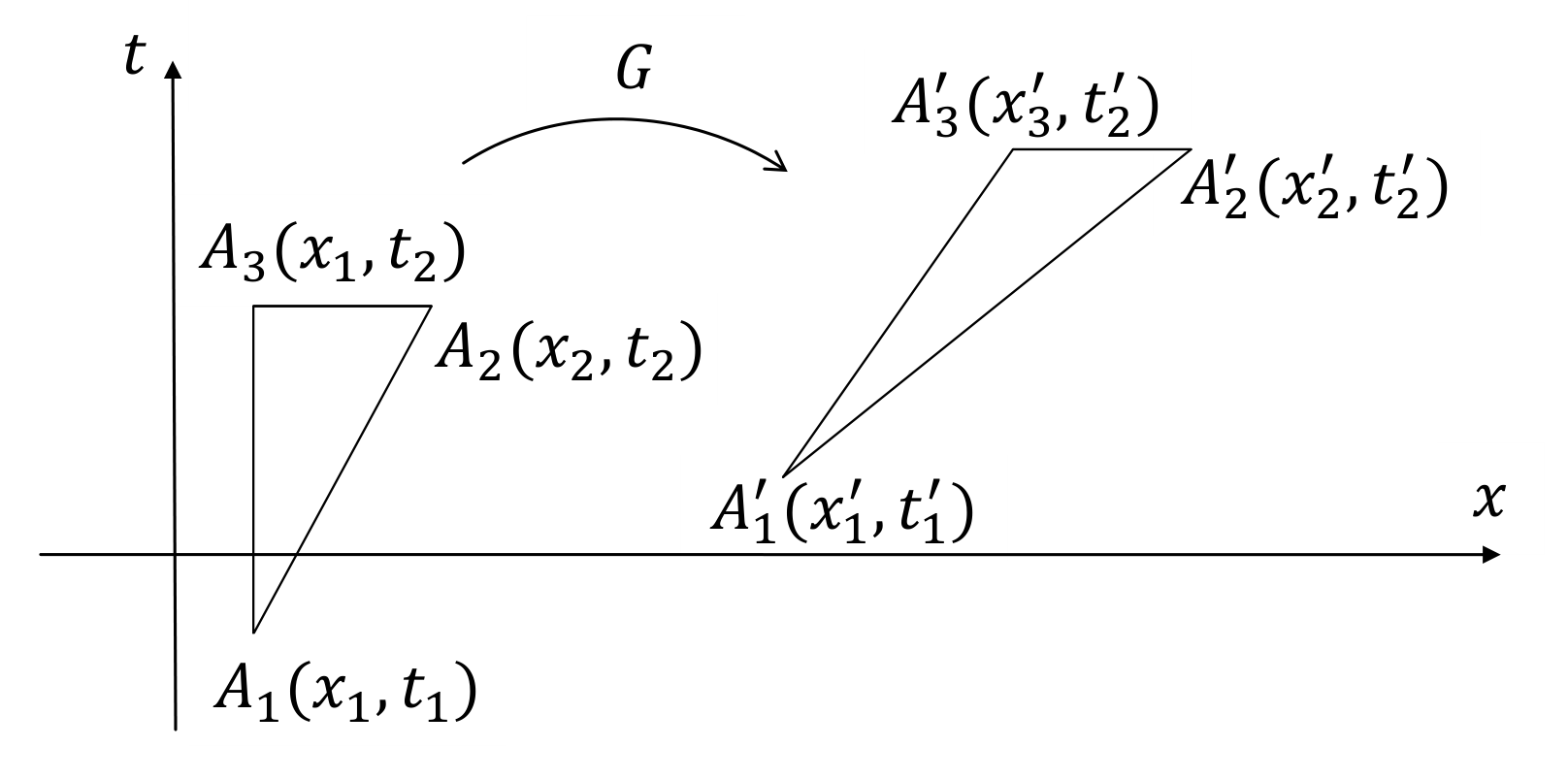}}
\caption{Galilean transformation of a plane} \label{galilejevo kretanje}}
\end{figure}

The group $SG(2)$ is called the \emph{group of Galilean motions} of a plane.
Although as a set it is equal to the product $SG_0(2)\times \mathbb R^2$, concerning the group structure given by \eqref{sg3} and \eqref{sg4},  $SG(2)$
is an example of a \emph{semi-direct product} $SG_0(2)\ltimes \mathbb R^2$ of the group of Galilean rotations $SG_0(2)$ that acts on the
group translations $\mathbb R^2$ of a plane.
Recall, similarly, the group of Euclidean motions $SE(2)$ (isometries that preserve orientation of the Euclidean plane) is isomorphic to a semi-direct
product $SO(2)\ltimes \mathbb R^2$ of the group of rotations $SO(2)$ that acts on the group of the translations $\mathbb R^2$ of a plane.

According to \eqref{sg1}, \eqref{sg2}, the group $(SG_0(2),\circ)$ is isomorphic to $(\mathbb R,+)$.
In addition, since $\det A(u)=1$, $SG_0(2)$ is a subgroup of the special linear group $SL(2)$ of the $2\times2$-matrixes with determinant 1:
\begin{equation}\label{sl2}
SL(2)=\Big\{ \mathrm A=\begin{pmatrix}
a & b  \\
c & d
\end{pmatrix} \big\vert\, \det\mathrm A=ad-bc=1,\, a,b,c,d\in\mathbb R\Big\}.
\end{equation}

Let $\Pi(\mathbf u,\mathbf v)$ denotes the oriented parallelogram spanned by $\mathbf u$ and $\mathbf v$, where the positive orientation is the orientation of
$\Pi(\mathbf e_1,\mathbf e_2)$.
Recall that the \emph{oriented area} $\mu$ of $\Pi(\mathbf u,\mathbf v)$  is given by the determinant:
\[
\mu(\Pi(\mathbf u,\mathbf v))=
\begin{vmatrix}
u_1 & v_1  \\
u_2 & v_2
\end{vmatrix}=u_1v_2-v_1u_2, \quad \mathbf u=u_1\mathbf e_1+u_2\mathbf e_2, \quad \mathbf v=v_1\mathbf e_1+v_2\mathbf e_2,
\]
and that the determinant of the matrix $\mathrm A$ can be seen  as the ratio of the oriented areas of parallelograms related by the corresponding linear mapping:
\[
\mu(\Pi(A\mathbf u,A\mathbf v))=\det\mathrm A \cdot  \mu(\Pi(\mathbf u,\mathbf v)), \qquad \mathbf u,\mathbf v\in\mathbb R^2.
\]

Thus, the Galilean transformation \eqref{G-u} preserves the standard area of the polygons in the plane $\mathbb R^2$.

The group of Galilean rotations $SG_0$ together with the group of Euclidean rotations (\emph{special orthogonal group})
$SO(2)=\{\mathrm R(\alpha)\,\vert\, \alpha\in [0,2\pi)\}$,
\begin{equation}\label{rotacije}
\begin{pmatrix}
v_1'\\
v_2'
\end{pmatrix}=
\mathrm R(\alpha)
\begin{pmatrix}
v_1\\
v_2
\end{pmatrix}=
\begin{pmatrix}
\cos\alpha & -\sin\alpha \\
\sin\alpha & \cos\alpha
\end{pmatrix}
\begin{pmatrix}
v_1\\
v_2
\end{pmatrix},
\end{equation}
are important ingredients of $SL(2)$.

\begin{theorem}[The Iwasawa decomposition of $SL(2)$]\label{probIWD}
Prove that for every matrix $\mathrm A\in SL(2)$ there exist unique rotations $\mathrm R(\alpha)\in SO(2)$ and $\mathrm A(u)\in SG_0(2)$, and a diagonal matrix
\[
\mathrm D(\sigma)=
\begin{pmatrix}
\sigma^{-1} & 0 \\
0 & \sigma
\end{pmatrix}, \qquad \sigma>0,
\]
such that
\begin{equation}\label{ID}
\mathrm A=\mathrm R(\alpha)\circ \mathrm A(u)\circ \mathrm D(\sigma).
\end{equation}
\end{theorem}

The solution of the problem is sketched at the beginning of Section \ref{vizualizacija}.

\begin{rem}\label{kolone}
The columns of the matrixes $\mathrm A\in SL(2)$ determine all positive oriented bases
$\mathbf v=\mathrm A\mathbf e_1=a\mathbf e_1+c\mathbf e_2$, $\mathbf w=\mathrm A\mathbf e_2=b\mathbf e_1+d\mathbf e_2$
of $\mathbb R^2$, such
that the area of the parallelogram $\Pi(\mathbf v,\mathbf w)$ is equal to 1.
Similarly, the columns of the matrixes $\mathrm R(\alpha)\in SO(2)$ define all positively oriented orthonormal bases  of $\mathbb R^2$, while
 the columns of the matrixes $\mathrm A(u)\in SG_0(2)$ define all Galilean bases  of $\mathbb R^2$.
\end{rem}

We also have the following interesting statement.

\begin{theorem}
Prove that the group of Galilean motions $SG(2)$ is isomorphic to the 3-dimensional Heisenberg group $H(3)$ defined by:
\[
H(3)=\Big\{
\begin{pmatrix}
1 & x & z \\
0 & 1 & y \\
0 & 0 & 1
\end{pmatrix} \big\vert\, x,y,z\in \mathbb R\Big\}.
\]
\end{theorem}

\section{From the Galilean to the Poincar\'e transformations.}\label{poenkareovSvet}

The Michelson–Morley experiment (1897) showed that the speed of light $c$ is the same in different inertial frames.
This significant and unexpected experiment was not in accordance with the addition of velocities in classical mechanics.
Also, it was observed that the Maxwell  equations of electrodynamics are not invariant under the Galilean transformations.
From that time it was clear that
Galilean transformations should be modified. The important step in this direction was performed by Hendrik Lorentz (1853--1928) in \cite{L}.
Poincar\'e corrected Lorentz's result, in particular in order to get that the set of transformations form a group (see  \cite{P2}).
He also obtained the new addition law for velocities \cite{P2,P3} and the underlying  pseudo-Euclidean geometry of the space-time \cite{P3}.
The constancy of the speed of light in all inertial systems was a starting point in Albert Einstein's (1879--1955) formulation of special relativity as well \cite{E}.
\footnote{While the legacy of Lorentz, Poincar\'e, and Einstein in the founding of special relativity is well known,
 we would like to mention here a contribution of Mileva Mari\'c (1875--1948) as well \cite{G}.}

Now, let us consider a 2-dimensional affine world without the Galilean structure. Because of the constancy of the speed of light we need
a different structure that will be formulated in the next section. We will adopt some of Poincar\'e's ideas and pass from
groups of affine transformations  to the geometry of space-time.

Here we assume that we have two inertial coordinate systems
$\mathbf I$ and $\mathbf I'$ (definition will be given also below) with coordinates $(x,t)$ and $(x',t')$ respectively.
Also, the spaces of simultaneous events $t=const$ and $t'=const$ within planes $(x,t)$ and $(x',t')$ have the Euclidean structures
as in the case of the Galilean geometry.

We are going to describe affine transformations
$P: \mathbb R^2\to \mathbb R^2$, $\mathbf x=\mathrm A\mathbf x'+\mathbf b$:
\begin{equation}\label{opsta}
\begin{pmatrix}
x\\
t
\end{pmatrix}=
\begin{pmatrix}
a_{11} & a_{12} \\
a_{21} & a_{22}
\end{pmatrix}
\begin{pmatrix}
x'\\
t'
\end{pmatrix}+
\begin{pmatrix}
b_1\\
b_2
\end{pmatrix}
\end{equation}
that satisfy the condition: the uniform motions with the velocity $\mathbf c'=\pm c\mathbf e_1'$ in the frame $\mathbf I'$ are mapped to the uniform motions
with the velocity $\mathbf c=\pm c\mathbf e_1$ in the frame $\mathbf I$.  Here $(b_1,b_2)$ are coordinates of the origin $O'$ in the reference frame $\mathbf I$
(see Fig. \ref{poenkareove transformacije}).
In addition, as in the case of Galilean transformations, it is natural to assume that the transformations \eqref{opsta} preserve the orientation
($\det\mathrm  A>0$).

\begin{theorem}
Let $\det\mathrm  A>0$. Prove that the affine transformation \eqref{opsta} maps the uniform motions
with the velocity $\mathbf c'=\pm c\mathbf e_1'$ in the frame $\mathbf I'$ to the uniform motions
with the velocity $\mathbf c=\pm c\mathbf e_1$ in the frame $\mathbf I$ if and only if
\begin{equation}\label{opste-A}
\mathrm A=\mathrm A_c(a,b)=
\begin{pmatrix}
a & b \\
\frac{b}{c^2} & a
\end{pmatrix},
\quad \text{where} \quad a,b\in\mathbb R, \, a^2-\frac{b^2}{c^2}>0.
\end{equation}
\end{theorem}

{\sc Hint.} For example, consider the two events $M(c,0)$ and $N(-c,1)$ that are connected with the origin $O'(0,0)$ with the uniform motions having the speed of light.

\begin{theorem} \label{Kc} Let
$
K_c=\big\{\mathrm A_c(a,b)\,\big\vert\, a,b\in\mathbb R, a^2-{b^2}/{c^2}>0\big\}.
$
Prove that $K_c$ is a commutative subgroup of the general linear group $GL(2)$ of $(2\times 2)$-invertible matrixes.
\end{theorem}

It is known that the intersection of two subgroups is also a subgroup.

\begin{theorem}\label{sve} Let $K^+_c=\{\mathrm A_c(a,b)\in K_c\,\vert\, a>0\}$.
\begin{itemize}
\item[(i)]   Prove that $K^+_c$ is a subgroup of $K_c$  (a connected component of $K_c$ that contains the identity matrix $\mathbf E$
\footnote{Let $G$ be a subgroup of the general linear group $GL(n)$.
The \emph{connected component} of $A=(a_{ij})\in G$
is the subset of the matrixes  $B=(b_{ij})\in G$,
such  that there exist a matrix curve $C(t)=(c_{ij}(t))\in G$ ($t\in [0,1]$),
where $c_{ij}(t)$ are continuous functions and $c_{ij}(0)=a_{ij}$, $b_{ij}=b_{ij}(1)$ ($i,j=1,\dots,n$).}).
\item[(ii)]  Prove that $SL(2)\cap K^+_c$ is the group
\[
SO^+_c(1,1)=\big\{\mathrm L_c(u)=
\kappa_c(u)\begin{pmatrix}
1 & u  \\
\frac{u}{c^2} & 1 \end{pmatrix} \big\vert\, u\in (-c,c), \, \kappa_c(u)=1/\sqrt{1-\frac{u^2}{c^2}}\,\big\}.
\]
\item[(iii)] Prove that
$K_c^+$ is isomorphic to the product of the groups $SO^+_c(1,1)$ and $\mathbb R^+$, where $\mathbb R^+$ is the group positive real numbers with respect to the
multiplication.
\end{itemize}
\end{theorem}

\begin{figure}[ht]
{\centering
{\includegraphics[width=8.5cm]{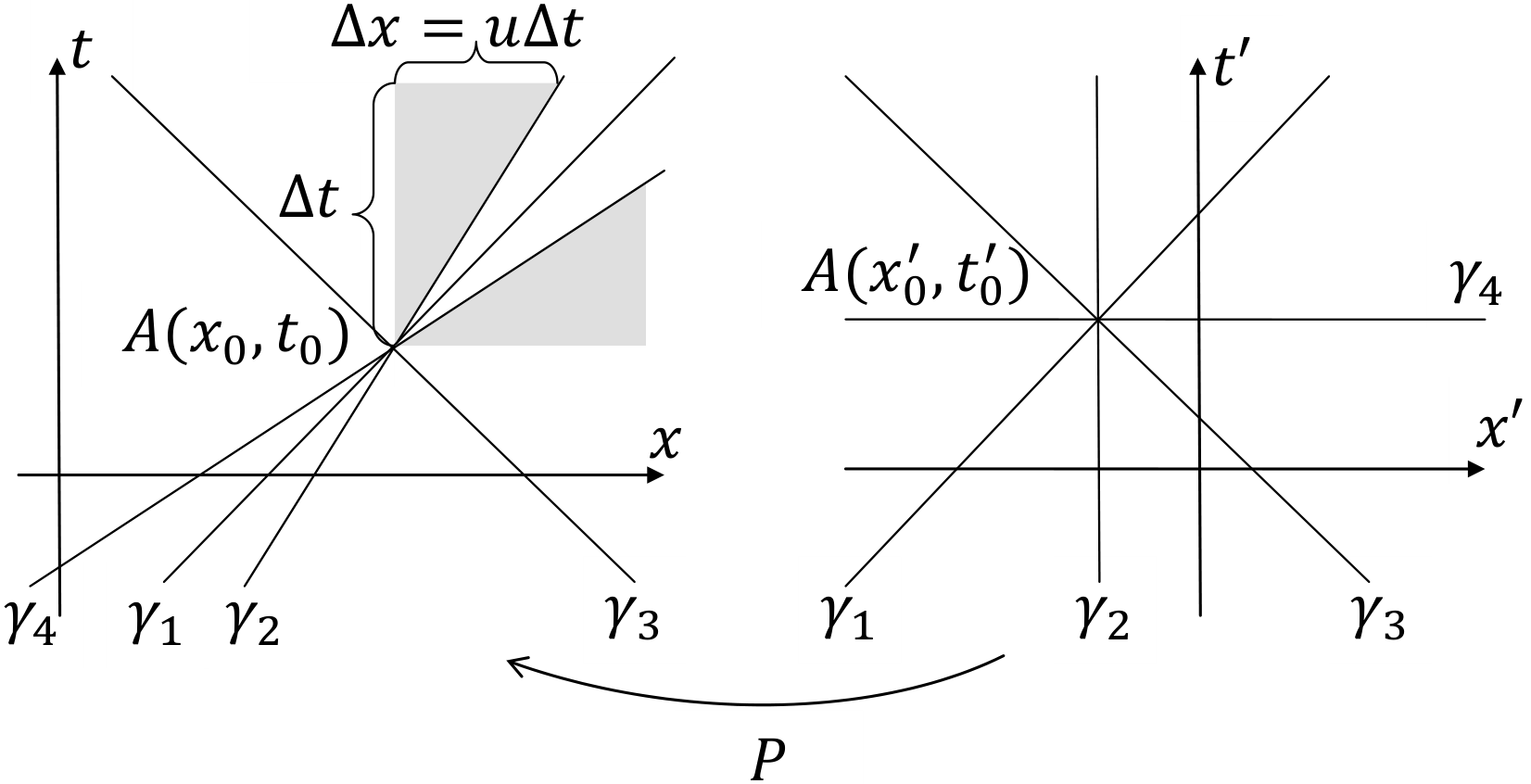}}
\caption{Four lines through an event $A$ in two coordinates systems related by an affine transformation of the form \eqref{opsta}, \eqref{opste-A},
in particular by a Poincar\'e transformation \eqref{poenkareove}.
$\gamma_1$ and $\gamma_3$ are uniform motions with the speed of light (we take $c=1$).
$\gamma_2=\{x'=x_0\}$ is at rest in the reference frame $\mathbf I'$.
$\gamma_4=\{t'=t_0'\}$ is the space of simultaneous events  in the frame $\mathbf I'$.
In both coordinates systems $\gamma_1$ and $\gamma_3$ are axis of symmetries for the union of intersecting lines $\gamma_2\cup\gamma_4$.
\label{poenkareove transformacije}}}
\end{figure}

The linear transformations defined by matrixes $\mathrm L_c(u)$ are obtained by Poincar\'e who named them the \emph{Lorentz transformations} \cite{P2, P3}.
On the other side, the affine transformations
$P_c(u,\mathbf b)=P_c(u,b_1,b_2): \mathbb R^2\to \mathbb R^2$:
\begin{equation}\label{poenkareove}
\begin{pmatrix}
x\\
t
\end{pmatrix}=
\kappa_c(u) \begin{pmatrix}
1 & u  \\
\frac{u}{c^2} & 1 \end{pmatrix}
\begin{pmatrix}
x'\\
t'
\end{pmatrix}+
\begin{pmatrix}
b_1\\
b_2
\end{pmatrix}, \quad u\in(-c,c), \, b_1,b_2\in\mathbb R
\end{equation}
are called \emph{Poincar\'e transformations}. We denote the Poincar\'e group by $SP^+_c(1,1)$.

\begin{theorem}\label{drugi}
Prove that the group multiplications and the inverses in the groups $SO^+_c(1,1)$ and $SP^+_c(1,1)$ are given by:
\begin{align}
\label{sp1}& \mathrm L_c(u)\circ \mathrm L_c(u')=\mathrm L\big(\frac{u+u'}{1+\frac{uu'}{c^2}}\big), \\
\label{sp2}& \mathrm L^{-1}_c(u)=\mathrm L_c(-u),\\
\label{sp3}& P_c(u,\mathbf b)\circ P_c(u',\mathbf b')=P_c\big(\frac{u+u'}{1+\frac{uu'}{c^2}},\mathrm L_c(u)\mathbf b'+\mathbf b\big),\\
\label{sp4}& P_c(u,\mathbf b)^{-1}=G(-u,-\mathrm L_c(-u)\mathbf b),
\end{align}
where $u\in (-c,c)$, $b_1,b_2\in\mathbb R$.
\end{theorem}

From now on, we consider the affine world $\mathcal A^2$ with the set of inertial frames related by Poincar\'e transformations.

As in the case of Galilean transformations we have the following kinematic interpretation of the parameter $u$:

\begin{theorem}\label{brzinaI'}
The motions that are at rest in the reference frame $\mathbf I'$ have the velocity $\mathbf u=u\mathbf e_1$
in the reference frame $\mathbf I$ (see Fig. \ref{poenkareove transformacije}).
\end{theorem}

Since $\vert u\vert<c$, the important consequence is that the
magnitude of velocity of an inertial frame $\mathbf I'$ is always smaller then
the speed of light.

\begin{theorem} [Addition of velocities in special relativity]
Let $u,v'\in (-c,c)$. Assume that a line $\gamma$ represents a uniform motion with the velocity
$\mathbf v'=v'\mathbf e_1'$ in the inertial frame  $\mathbf I'$, and that
the velocity of the frame $\mathbf I'$ with respect to $\mathbf I$ is $\mathbf u=u\mathbf e_1$.
Then the line $\gamma$ represents a uniform motion with the velocity
$\mathbf v=v\mathbf e_1$ in the inertial frame  $\mathbf I$, where
\[
v=\frac{u+v'}{1+\frac{uv'}{c^2}}.
\]
\begin{itemize}
\item[(i)]  Prove the statement directly, by using the definition $v={\Delta x}/{\Delta t}$ and the Poincar\'e transformation \eqref{poenkareove}.
\item[(ii)] Prove the statement by using Problem \ref{drugi}.
\end{itemize}
\label{sabiranje}\end{theorem}

\begin{rem}
Let us observe that the Poincar\'e transformation $P_c(u,\mathbf b)$
tends to the Galilean transformation $G(u,\mathbf b)$ as $c$ tends to infinity.
Moreover, the equations \eqref{sp3} and \eqref{sp4} tend to \eqref{sg3} and \eqref{sg4}
as $c$ tends to infinity. We say that the group of Galilean motions $SG(2)$ is the limit
of the Poincar\'e group $SP^+_c(1,1)$ as $c$ tends to infinity.
\end{rem}

\begin{rem}\label{invarijantnost}
Concerning the geometrical structure, Poincar\'e also realised that
the Lorentz transformations $\mathbf x'=\mathrm L_c(u)\mathbf x$ can be defined from the condition of the
preservation of the following indefinite quadratic form in the space-time (see \cite{P3}):
\begin{equation}\label{QF}
(x_2-x_1)^2-c^2(t_2-t_1)^2=(x_2'-x_1')^2-c^2(t_2'-t_1')^2.
\end{equation}
The role of the pseudo-Euclidean geometry is further developed by Hermann Min\-kowski (1864--1909).
\end{rem}

\subsection{1-parametric deformation of the Poincar\'e group.}\label{dodatak}
Instead of the presentation of the matrix $\mathrm L_c(u)$ with a kinematic parameter $u$ and the group structure given by \eqref{sp1}, \eqref{sp2}, it is
convenient to set (usually with $c=1$)
\[
\mathrm L_c^s=\mathrm L_c(u)\vert_{u=c\tanh s}=\begin{pmatrix}
  \cosh s & c\sinh s  \\
\frac{1}{c}\sinh s &  \cosh s  \end{pmatrix}, \quad s \in\mathbb R.
\]

Then
$\mathrm L_c^{s}\circ \mathrm L_c^{s'}=\mathrm L_c^{s+s'}$,
$(\mathrm L_c^s)^{-1}=\mathrm L_c^{-s}$ ($s,s'\in\mathbb R$),
and the addition of velocities is equivalent to the identity
\[
\tanh(s+s')=\frac{\tanh s+\tanh s'}{1+\tanh s\tanh s'}.
\]

Recall that a smooth mapping $\phi\colon  s\mapsto g^s$ of $\mathbb R$ to a Lie group $(G,\cdot)$
is a \emph{one-parametric subgroup} if
$g^0$ is the neutral in $G$ and $g^s\cdot g^{s'}=g^{s+s'}$ ($s,s'\in\mathbb R$).

Let $\mathfrak g=T_e G$ be the Lie algebra of $G$.
The one-parametric subgroup $\phi$ determines the \emph{exponential mapping} $\exp\colon\mathfrak g\to G$,
$\phi(s)=\exp(s\xi)$, where $\xi$ is tangent to the curve $\phi(s)$ at $0$.
For the matrix groups, the exponential mapping coincides with the usual exponential mapping of matrixes.

For example, the groups $SO^+_c(1,1)=\{\mathrm L_c^s\,\vert\,s\in\mathbb R\}$, $SO(2)= \{\mathrm R(\alpha)\,\vert\,\alpha\in\mathbb R\}$,
and $SG_0(2)= \{\mathrm A(u)\,\vert\,u\in\mathbb R\}$ are the images of one-parametric subgroups in $SL(2)$.

Let us consider $K^+_c$, the connected component of the group of linear transformations that ensure that the speed of
the light is the same in coordinate transformation between inertial systems given by  \eqref{opsta}, \eqref{opste-A}
(see Problem \ref{sve} and Fig. \ref{poenkareove transformacije}). The Lie algebra $\mathfrak k^+_c$ of $K^+_c$ is generated by
\[
\mathbf E=
\begin{pmatrix}
1 & 0  \\
0 & 1
\end{pmatrix}
\quad  \text{and} \quad
\mathbf F_c=
\begin{pmatrix}
0 & c  \\
\frac{1}{c} & 1
\end{pmatrix}.
\]
From $[\mathbf E,\mathbf F_c]=0$, we get
$\exp(s(a\mathbf E+b\mathbf F_c))=\exp(sa\mathbf E)\circ\exp(sb\mathbf F_c)$ ($a,b\in\R$).
Next, since $\exp(sa\mathbf E)=\exp(sa)\mathbf E$ and $\exp(sb\mathbf F_c)=\mathrm L_c^{sb}$,
we obtain the following statement:

\begin{prop}
One-parametric subgroups of $K^+_c$ are of the form
\[
\phi_{a,b}\colon \mathbb R\rightarrow K^+_c, \qquad \phi_{a,b}(s)=\exp({as})\mathrm L_c^{bs}, \qquad s\in\mathbb R,
\]
where $a, b$ are real parameters.
\end{prop}

Let
$\Phi_{a,b}=\phi_{a,b}(\mathbb R)$
 be the corresponding Lie subgroups of $K^+_c$.
In particular, $\Phi_{0,1}$ is the group of Lorentz transformation $SO^+_c(1,1)$  and $\Phi_{1,0}$ is the
subgroup of diagonal matrixes $\{\exp(s)\mathbf E\,\vert\, s\in\mathbb R\}$ defining homothetic transformations of $\mathbb R^2$ centered at the origin.
We have a decomposition of $K_c^+$ on 1-dimensional connected Lie subgroups
\begin{equation}\label{dekompozicija}
K_c^+=SO^+_c(1,1)\cup_{a\ne 0} \Phi_{a,1} \cup \Phi_{1,0}
\end{equation}
that mutually intersects only at the unit matrix $\mathbf E$. While the group of homothetic transformations $\Phi_{1,0}$ is not the candidate,
besides $SO^+_c(1,1)$ all others 1-dimensional subgroups
\[
\Phi_{a,1}=\{\exp({as})\mathrm L_c^{s}\,\vert\, s\in\mathbb R\}=\{\exp({a\,\mathrm{artanh}(u/c)})\mathrm L_c(u)\,\vert\, u\in (-c,c)\},
\]
$a\ne 0$, could be also considered as candidates for a linear part of the affine transformations \eqref{opsta}  between inertial frames $\mathbf I$ and $\mathbf I'$.

Let $\ell_{c,a}(u)=\exp({a\,\mathrm{artanh}(u/c)})$.
Consider the affine transformations
$P_{c,a}(u,\mathbf b): \mathbb R^2\to \mathbb R^2$ defined by:
\begin{equation}\label{H-poenkareove}
\begin{pmatrix}
x\\
t
\end{pmatrix}=
\ell_{c,a}(u)\kappa_c(u) \begin{pmatrix}
1 & u  \\
\frac{u}{c^2} & 1 \end{pmatrix}
\begin{pmatrix}
x'\\
t'
\end{pmatrix}+
\begin{pmatrix}
b_1\\
b_2
\end{pmatrix},
\end{equation}
$u\in(-c,c)$, $\mathbf b=(b_1,b_2)\in\mathbb R^2$. Denote the corresponding group by $SP^+_{c,a}(1,1)$.

\begin{theorem}\label{cudan}
Prove Problems \ref{brzinaI'} and \ref{sabiranje} (addition of velocities in special relativity) where instead of the Poincar\'e transformation \eqref{poenkareove}
one uses the transformation \eqref{H-poenkareove} from the group $SP^+_{c,a}(1,1)$.
\end{theorem}

Note that both Poincar\'e and Einstein considered the linear part of the transformations
\eqref{H-poenkareove} with an unknown function $\ell(u)$  (see \cite{E, P1, P2}).
Here, in a sense, we follow Poincar\'e's approach: $\ell(u)\equiv 1$ is the simplest function that leads to the fact that the corresponding set of transformations forms a group, which, in addition, preserves the quadratic form \eqref{QF}.

Curiously, we did not find
the decomposition \eqref{dekompozicija} and the family of 1-parametric deformations $SP^+_{c,a}(1,1)$ of the Poincar\'e group in the literature.
Problem \ref{cudan} indicates that one should be very careful in derivation of Lorentz transformation.

\section{Pseudo-Euclidean 2-dimensional world.}\label{pseudoSvet}


Consider an affine plane $\mathcal A^2$  and assume that the associated vector space $\mathbb V$ is endowed with
a symmetric, bilinear, nondegenerate  form
$\langle \cdot,\cdot\rangle\colon \mathbb V\times \mathbb V \to \mathbb R$:
\begin{itemize}
\item[(a)] $\langle \mathbf u,\mathbf v\rangle=\langle \mathbf v,\mathbf u\rangle\quad $ ($\mathbf u,\mathbf v\in\mathbb V$);
\item[(b)] $\langle x_1\mathbf u_1+x_2\mathbf u_2,\mathbf v\rangle=x_1\langle \mathbf u_1,\mathbf v\rangle+
x_2\langle \mathbf u_2,\mathbf v\rangle\quad $ ($\mathbf u_1,\mathbf u_2,\mathbf v\in\mathbb V$, $x_1,x_2\in\mathbb R$);
\item[(c)] If $\langle \mathbf u,\mathbf v\rangle=0$ for all $\mathbf v\in\mathbb V$, then $\mathbf u=0$.
\end{itemize}

There are three possible situations:
\begin{itemize}
\item[(i)] $\langle \mathbf u,\mathbf u\rangle>0$, for all non-zero vectors $\mathbf u\in\mathbb V$, that is $\langle \cdot,\cdot\rangle$ is the Euclidean scalar product.
We say that $\langle \cdot,\cdot\rangle$ is a \emph{scalar product of signature} $(2,0)$;
\item[(ii)] $\langle \mathbf u,\mathbf u\rangle<0$, for all non-zero vectors $\mathbf u\in\mathbb V$, that is $-\langle \cdot,\cdot\rangle$ is the Euclidean scalar product.
We say that We say that $\langle \cdot,\cdot\rangle$ is a \emph{scalar product of signature} $(0,2)$;
\item[(iii)] There exist $\mathbf u,\mathbf v\in\mathbb V$, such that $\langle \mathbf u,\mathbf u\rangle>0$, $\langle \mathbf v,\mathbf v\rangle<0$.
We say that $\langle \cdot,\cdot\rangle$ is a \emph{scalar product of signature} $(1,1)$.
\end{itemize}

\begin{theorem}
Prove that indeed only the above three situations can appear. For example, the situation where all non-zero vectors $\mathbf u$ satisfy either
$\langle \mathbf u,\mathbf u\rangle>0$ or $\langle \mathbf u,\mathbf u\rangle=0$ is impossible.
\end{theorem}

{\sc Hint.} Let $\langle \mathbf u,\mathbf u\rangle>0$ and $\langle\mathbf v,\mathbf v\rangle=0$. Prove that $\langle \mathbf u,\mathbf v\rangle\ne 0$ and consider the function
$f(x)=\langle x\mathbf u-\mathbf v,x\mathbf u-\mathbf v\rangle$, $x\in\mathbb R$.

\

Let us consider a scalar product of signature $(1,1)$.
As in the usual Euclidean plane, the scalar product defines the pseudo-Euclidean quadratic form
in the affine plane:
\[
\delta(A,B)=\langle \overrightarrow{AB}, \overrightarrow{AB} \rangle.
\]

We say that $(\mathcal A^2, \delta)$ (respectively $(\mathbb V,\langle\cdot,\cdot\rangle)$)
is a \emph{pseudo-Euclidean} affine space (respectively, a vector space) of signature $(1,1)$.

\begin{dfn}
Let $\mathbf v\in\mathbb V$, $\mathbf v\ne 0$. Then
\begin{itemize}
\item[(i)] If $\langle \mathbf v,\mathbf v\rangle>0$ then $\mathbf v$ is called a \emph{space-like vector};
\item[(ii)] If $\langle \mathbf v,\mathbf v\rangle=0$ then $\mathbf v$ is called a \emph{light-like vector};
\item[(iii)] If $\langle \mathbf v,\mathbf v\rangle<0$ then $\mathbf v$ is called a \emph{time-like vector}.
\end{itemize}
\end{dfn}

The segment $[AB]$ (or the line $AB$) is \emph{space-like}, \emph{time-like}, or the \emph{light-like}, depending of the type of the vector  $\overrightarrow{AB}$.
Note that the definition of time-like and space-like vectors and lines can be interchanged. Both notations are used in the literature.

A basis $\mathbf e_1, \mathbf e_2$ is \emph{pseudo-othonormal} if
\[
\langle \mathbf e_1,\mathbf e_1\rangle=1, \qquad \langle \mathbf e_2,\mathbf e_2\rangle=-1, \qquad \langle \mathbf e_1,\mathbf e_2\rangle=0.
\]

The construction of a pseudo-orthonormal basis is similar like in
the Euclidean geometry. Take two vectors $\mathbf u,\mathbf
v\in\mathbb V$, such that $\langle \mathbf u,\mathbf u\rangle>0$,
$\langle \mathbf v,\mathbf v\rangle<0$. Then $\mathbf e_1=\mathbf
u/\sqrt{\langle \mathbf u,\mathbf u\rangle}$ satisfies $\langle
\mathbf e_1,\mathbf e_1\rangle=1$. Further, let $\mathbf
v'=\mathbf v-\langle \mathbf v,\mathbf e_1\rangle \mathbf e_1$.
Then $\mathbf v'$ and $\mathbf e_1$ are \emph{pseudo-orthogonal}:
\[
\langle \mathbf v',\mathbf e_1\rangle=\langle\mathbf v-\langle \mathbf v,\mathbf e_1\rangle\mathbf e_1,\mathbf e_1\rangle=
\langle \mathbf v,\mathbf e_1\rangle-\langle \mathbf v,\mathbf e_1\rangle=0.
\]
Finally, we take $\mathbf e_2=\mathbf v'/\sqrt{\vert\langle \mathbf v',\mathbf v'\rangle\vert}$.

Let $\mathbf v=v_1\mathbf e_1+v_2\mathbf e_2$, where $\mathbf e_1,
\mathbf e_2$ is a pseudo-othonormal basis. Then the space-like,
light-like, and time-like vectors represent the regions of
$\mathbb V$ determined by the equations
\[
\mathbb V_S\colon v_1^2-v_2^2>0, \quad \mathbb V_L\colon v_1^2-v_2^2=0, \,v_1^2+v_2^2\ne 0, \quad \mathbb V_T\colon v_1^2-v_2^2<0.
\]
The subset $\mathbb V_L$ is called the \emph{light-like cone}. In the two-dimensional case this is the union of two one-dimensional subspaces without the origin.
Note that the region of time-like vectors $\mathbb V_T$ and the light-like cone have two components:
\begin{align*}
&\mathbb V_T^+\colon v_1^2-v_2^2<0, v_2>0, \qquad \mathbb V_T^-\colon v_1^2-v_2^2<0, v_2<0, \\
&\mathbb V_L^+\colon v_1^2-v_2^2=0, v_2>0, \qquad \mathbb V_L^-\colon v_1^2-v_2^2=0, v_2<0.
\end{align*}

Once we fix the decompositions $\mathbb V_T=\mathbb V_T^+\cup \mathbb V_T^-$, $\mathbb V_L=\mathbb V_L^+\cup \mathbb V_L^-$
 we call the  vectors in $\mathbb V^+=\mathbb V^+_T\cup\mathbb V^+_L$ the \emph{positive vectors} and
 the  vectors in $\mathbb V^-=\mathbb V^-_T\cup\mathbb V^-_L$ the \emph{negative vectors}. They define positive and
 negative orientations for the time-like and light-like lines.

It is clear that the sum $\mathbf u+\mathbf v$ of a two positive time-like vectors $\mathbf u$ and $\mathbf v$ is a positive time-like vector.

\begin{theorem}\label{stavNejednakosti}
 Let $\mathbf v=v_1 \mathbf e_1+v_2\mathbf e_2$, $\mathbf v=v_1 \mathbf e_1+v_2\mathbf e_2$ be nonproportional positive time-like or light-like vectors.
Prove:
\begin{itemize}
\item[(i)] $\vert\langle \mathbf u,\mathbf v\rangle\vert=u_2v_2-u_1v_1>\sqrt{\langle \mathbf u,\mathbf u\rangle\langle\mathbf v,\mathbf v\rangle}$;
\item[(ii)]$\sqrt{\vert\langle \mathbf u+\mathbf v,\mathbf u+\mathbf v\rangle\vert}>
\sqrt{\vert\langle\mathbf u,\mathbf u\rangle\vert}+\sqrt{\vert\langle\mathbf v,\mathbf v\rangle\vert}$.
\end{itemize}
\end{theorem}

Now we are ready to define a 2-dimensional affine world of special relativity.
This is a pseudo-Euclidean affine space $\mathcal A^2$ of signature $(1,1)$.
The \emph{uniform motions} in the affine world $\mathcal A^2$ are time-like and light-like lines with positive orientations.
For $A\in\mathcal A^2$ we define the region of \emph{future events} $A^+=A+\mathbb V^+$ (the events that can be influenced by $A$)
and the region of \emph{past events} $A^-=A+\mathbb V^-$ (the events that may affect $A$), see Fig. \ref{buduci dogadjaji}.
Also, $A+\mathbb V_L$ is called the \emph{light-like cone at} $A$.

It is clear that if $B\in A^+$, then the region of future events of $B$ is a subset of the region of future events of $A$: $B^+\subset A^+$
(see Fig. \ref{diletacija-kontrakcija}).

Let $B\in A^-\cup A^+$. Then the \emph{proper time} between events $A$ and $B$ is defined by
\[
\tau(A,B)=\frac1{c}\sqrt{\vert\delta(A,B)\vert}=\frac1{c}{\sqrt{\vert\langle\overrightarrow{AB},\overrightarrow{AB}\rangle\vert}},
\]
where, as above, $c$ denotes the speed of light. The physical interpretation of the proper time $\tau(A,B)$ is that it is a time passed for an observer who moves
from the event $A$ to the event $B$ along the line $AB$ (in the case $B\in A^+$). If $B$ belongs in the light-like cone of $A$, then $\tau(A,B)=0$.

As a direct corollary of Problem \ref{stavNejednakosti} we get the \emph{inverse triangle inequality}
(compare with the \emph{triangle equality} $l(A,B)+l(B,C)=l(A,C)$ in the Galilean world).

\begin{prop}[Twin paradox] Let $A,B,C$ be three noncollinear events, such that $B\in A^+$ and $C\in B^+$.
Then the proper time between the events $A$ and $C$ is greater then the sum of the proper times between the events $A$ and $B$ and $B$ and $C$:
\[
\tau(A,C)>\tau(A,B)+\tau(B,C).
\]
\end{prop}

We also fix an orientation of the vector space $\mathbb V$, that is the associated orientation of the affine world $\mathcal A^2$.

\begin{dfn}
An event $O\in \mathcal A^2$ and a pseudo-orthonormal positively
oriented basis $\mathbf e_1,\mathbf e_2$ with $\mathbf e_2$ being
positive time-like vector define an \emph{inertial reference
frame} $\mathbf I=[O,\mathbf e_1\mathbf e_2]$ in the
pseudo-Euclidean world $\mathcal A^2$. The inertial coordinates
$(x,t)\in \mathbb R^2$ of an event $X\in\mathcal A^2$ with respect
the inertial frame $\mathbf I$ are defined by
\[
\overrightarrow{OX}=x\mathbf e_1+ct\mathbf e_2,
\]
i.e., the coordinates $(x,t)$ are the affine coordinates with
respect to the basis $\mathbf e_1,c\mathbf e_2$ of $\mathbb V$.
\end{dfn}

In inertial coordinates $(x,t)$ we have (see Fig. \ref{buduci dogadjaji})
\begin{align*}
& \overrightarrow{A(x_1,t_1),B(x_2,t_2)}=(x_2-x_1)\mathbf e_1+c(t_2-t_1)\mathbf e_2,\\
& A(x,t)+v_1\mathbf e_1+v_2\mathbf e_2=\big(x+v_1,t+\frac{v_2}{c}\big),
\end{align*}
and the pseudo-Euclidean quadratic form is given by:
\begin{equation}\label{quadraticForm}
\delta(A(x_1,t_1),B(x_2,t_2))=(x_2-x_1)^2-c^2(t_2-t_1)^2.
\end{equation}

Note that it is also usual to write affine coordinates in the form $(x,ct)$.

\begin{figure}[ht]
{\centering
{\includegraphics[width=7.5cm]{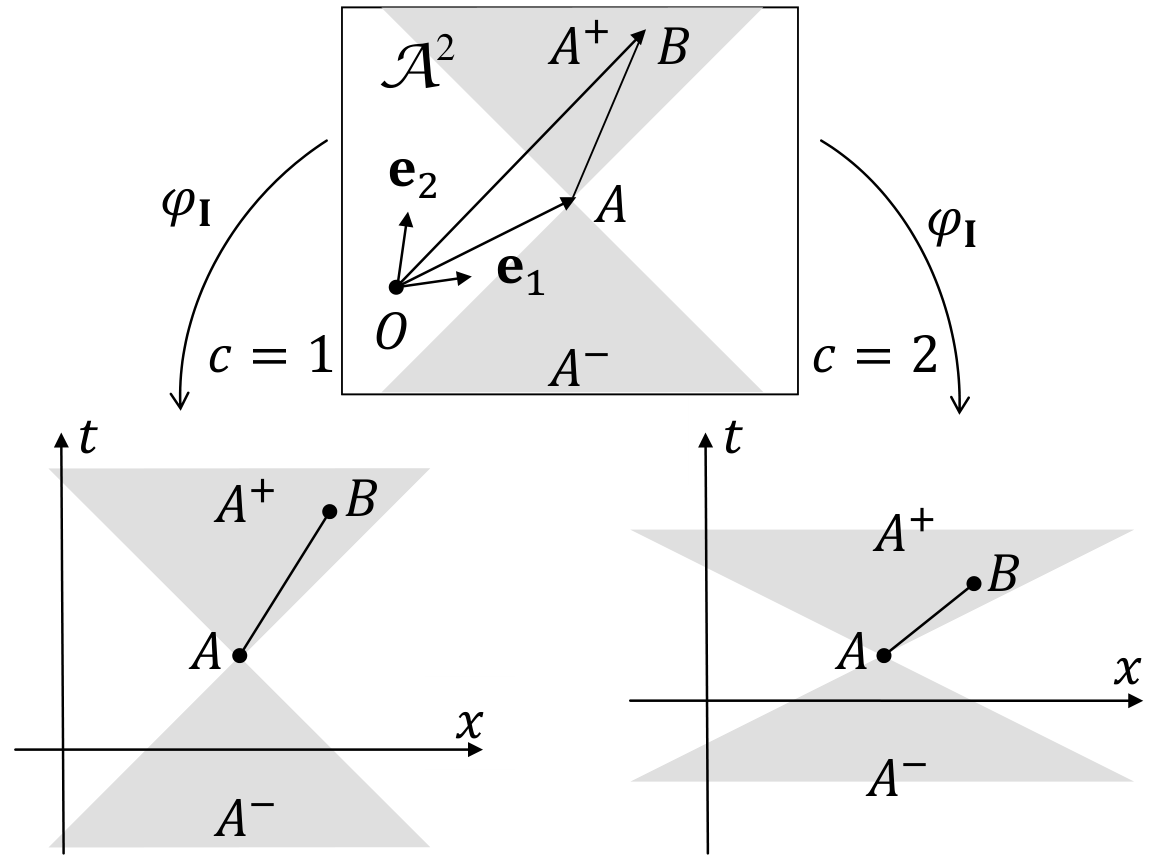}}
\caption{The region of future and past events of $A$ in the inertial reference frame $\mathbf I=[O,\mathbf e_1,\mathbf e_2]$ with the speed of light equals $c=1$ and $c=2$.
Note that if we use the system of units with meters for $x$-axis and seconds for $t$-axis, and the value $c=3\cdot 10^8 m/s$, then the space of future events would not be
distinguished from the upper half plane.
\label{buduci dogadjaji}}}
\end{figure}

In the Galilean world, the family of spaces of simultaneous events does not depend of the chosen inertial system.
Here, the spaces of simultaneous events $t=const$ in the frame $\mathbf I$ depend of $\mathbf e_2$:
they have direction $\mathbf e_2^\perp$, where
$
\mathbf e_2^\perp=\{\mathbf v\in \mathbb V\,\vert\, \langle \mathbf e_2,\mathbf v\rangle=0\}
$
is the \emph{pseoudo-orthogonal complement} of $\mathbf e_2$.
Since the non-zero vectors in $\mathbf e_2^\perp$ are space-like, the restriction of the scalar product $\langle \cdot,\cdot\rangle$ to
$\mathbf e_2^\perp$ defines the Euclidean scalar product.

Therefore, in the chosen inertial reference frame $\mathbf I$
we can consider the Galilean structure:
\begin{align*}
& l(A(x_1,t_1)),B(x_2,t_2))=t_2-t_1 ,\\
& l(A,B)=0  \quad \Rightarrow \quad \rho(A(x_1,b_1),B(x_2,t_2))=\sqrt{\delta(A,B)}=\vert x_2-x_1\vert.
\end{align*}

We say that the event $A(x_1,t_1)$ happened before the event $B(x_2,t_2)$ in the inertial reference frame $\mathbf I$ if
$l(A,B)=t_2-t_1>0$.

Further, consider a uniform motion given by a time-like or light-like line $\gamma$ and let $A(x_1,t_1)$ and $B(x_2,t_2)$ be
two events on $\gamma$ ($t_1<t_2$).
The Galilean time $l(A,B)=t_2-t_1$ between the events $A$ and $B$ coincides
with the proper time  $\tau(A,B)$ only if the uniform motion along the line $AB$ is at rest in the reference frame $\mathbf I$.
Otherwise, we have  (so called \emph{dilatation of time}):
\begin{align*}
\tau(A,B)&=\frac{1}{c}\sqrt{c^2(t_2-t_1)^2-(x_2-x_1)^2}\\
&=(t_2-t_1)\sqrt{1-(v/c)^2}<(t_2-t_1) =l(A,B),
\end{align*}
where $v=\Delta x/\Delta t=(x_2-x_1)/(t_2-t_1)$.
Also, apart of the space-like velocity $\mathbf v=v\mathbf e_1$, for a time-like line $\gamma$ ($v\in(-c,c)$),
we define the \emph{time-like velocity}:
\[
\mathbf V=\frac{\overrightarrow{AB}}{\tau(A,B)}=
\frac1{\sqrt{1-(v/c)^2}}\frac{\overrightarrow{AB}}{l(A,B)}=v\kappa_c(v)\mathbf e_1+c\kappa_c(v)\mathbf e_2.
\]
The time-like velocity $\mathbf V$ is the normalized ($\langle \mathbf V,\mathbf V\rangle=-c^2$) direction of $\gamma$.

\begin{theorem}  Let $\mathbf I=[O,\mathbf e_1\mathbf e_2]$ and $\mathbf I'=[O',\mathbf e_1',\mathbf e_2']$
be inertial reference frames  with
affine coordinates $(x,t)$ and $(x',t')$ defined above. Prove:
\begin{itemize}
\item[(i)] Consider the vector $\mathbf w\in\mathbb V$ in the bases $\mathbf e_1\mathbf e_2$ and $\mathbf e_1',\mathbf e_2'$:
\[
\mathbf w =w_1\mathbf e_1+w_2\mathbf e_2=w_1'\mathbf e_1'+w_2'\mathbf e_2.
\]
Then the components $w_1,w_2$ and $w_1',w_2'$ are related by a Lorentz transformation of the form
\begin{equation}\label{lorencove*}
\begin{pmatrix}
w_1\\
w_2
\end{pmatrix}=\tilde{\mathrm  L}_c(u)
\begin{pmatrix}
w_1'\\
w_2'
\end{pmatrix}
=
\kappa_c(u) \begin{pmatrix}
1 & \frac{u}{c}  \\
\frac{u}{c} & 1 \end{pmatrix}
\begin{pmatrix}
w_1'\\
w_2'
\end{pmatrix}.
 \end{equation}
In particular, the pseudo-orthonormal bases are related by:
\begin{align*}
&\mathbf e_1'=\kappa_c(u)\big(\mathbf e_1+\frac{u}{c}\mathbf e_2\big), \qquad \mathbf e_1'=\kappa_c(u)\big(\frac{u}{c}\mathbf e_1+\mathbf e_2\big),\\
&\mathbf e_1=\kappa_c(u)\big(\mathbf e_1'-\frac{u}{c}\mathbf e_2'\big), \qquad \mathbf e_1=\kappa_c(u)\big(-\frac{u}{c}\mathbf e_1'+\mathbf e_2'\big).
\end{align*}

\item[(ii)] The coordinates $(x,t)$ and $(x',t')$ are related by a Poincar\'e transformation \eqref{poenkareove},
where $\mathbf b=(b_1,b_2)$ are the coordinates of $O'$ in variables $(x,t)$. The parameter $u$ is the
magnitude of the velocity $\mathbf u=u\mathbf e_1$
of a uniform motion $\gamma$ that is at rest in the reference frame $\mathbf I'$.
Equivalently, in the coordinates $(x,ct)$ and $(x',ct')$,
the Poincar\'e transformations $\tilde P_c(u,\mathbf b)\colon \mathbb R^2\to \mathbb R^2$ takes the form:
\footnote{In the literature,  both linear maps, given by
$\mathrm L_c(u)$ and by
$\tilde{\mathrm L}_c(u)$, are called Lorentz transformations. Also, both affine transformations $P_c(u,\mathbf b)$ and $\tilde P_c(u,\mathbf b)$ are called
Poincar\'e transformations.}
\begin{equation}\label{poenkareove*}
\begin{pmatrix}
x\\
ct
\end{pmatrix}=
\kappa_c(u) \begin{pmatrix}
1 & \frac{u}{c}  \\
\frac{u}{c} & 1 \end{pmatrix}
\begin{pmatrix}
x'\\
c t'
\end{pmatrix}+
\begin{pmatrix}
b_1\\
c b_2
\end{pmatrix}.
\end{equation}
\end{itemize}
\end{theorem}

The solution can be deduced from the proof of Proposition \ref{ortogonalnaGrupa} given below.

\begin{rem} \label{DrugaPrimedba}
In general, a motion (world line) in the affine world $\mathcal A^2$ is a curve
\begin{equation}\label{kretanje*}
\gamma=\{\gamma(t)=(x(t),t)\,\vert\, t\in\R\},
\end{equation}
where $x(t)$ is a smooth function of the time $t$ in the chosen inertial reference frame $\mathbf I$ with  the space-like velocity
$\mathbf v(t)=v(t) \mathbf e_1$ that satisfies $v(t)={dx}/{dt}\in[-c,c]$, $t\in\mathbb R$.
Let $A=\gamma(t_1)$, $B=\gamma(t_2)$, $t_1<t_2$. The \emph{proper time between events $A$ and $B$ along the motion} $\gamma$ is defined by
\[
\tau(A,B,\gamma)=\int_{t_1}^{t_2}\sqrt{1-(v(t)/c)^2} dt.
\]
It can be proved that the proper time between events $A$ and $B$ along the motion $\gamma$ does not depend on the chosen inertial system $\mathbf I$.
Also, the \emph{time-like velocity}
\[
\mathbf V(C)=v(t)\kappa_c(v(t))\mathbf e_1+c\kappa_c(v(t))\mathbf e_2, \qquad C=\gamma(t)\in\gamma,
\]
does not depend on inertial frame $\mathbf I$. The components of $\mathbf V(C)$ in two reference frames are related by the Lorentz transformation
\eqref{lorencove*}. Here we can formulate a \emph{general twin paradox}: Let $B\in A^+$. The proper time $\tau(A,B)$ along the uniform motion from $A$ to $B$
is greater that the proper time $\tau(A,B,\gamma)$ along arbitrary non-uniform motion \eqref{kretanje*}.
\end{rem}

Note that if $B\in A^+$, then $l(A,B)$ is positive for all inertial frames $\mathbf I$ (see Fig. \ref{buduci dogadjaji}).
However, if ${AB}$ is a space-like line, then one can easily
construct inertial reference frames where $A$ and $B$ are simultaneous events, where the event $A$ happened before the event $B$, and where
the event $A$ happened after the event $B$.
Moreover, we have:

\begin{theorem} [Contraction of length]
Consider uniform motions given by two parallel time-like lines $\gamma_1$ and $\gamma_2$.
Let $A_1$  and $A_2$ be two arbitrary simultaneous events, $A_1\in\gamma_1$, $A_2\in\gamma_2$, with respect to the
frame $\mathbf I'=[O',\mathbf e_1',\mathbf e_2']$ where $\gamma_1$ and $\gamma_2$ are at rest, and let $B_1$ and $B_2$ be two arbitrary simultaneous
events,  $B_1\in\gamma_1$, $B_2\in\gamma_2$, with respect to the
frame $\mathbf I=[O,\mathbf e_1,\mathbf e_2]$ where $\gamma_1$ and $\gamma_2$ have the velocity $\mathbf u=u\mathbf e_1$
(see Fig. \ref{diletacija-kontrakcija}).
Then
the Euclidean distances $\rho(B_1,B_2)$ and $\rho'(A_1,A_2)$ with respect to the frames $\mathbf I$ and $\mathbf I'$ are related by
\[
\rho(B_1,B_2)=\rho'(A_1,A_2)\sqrt{1-(u/c)^2}<\rho'(A_1,A_2),
\]
that is,
$
\delta(B_1,B_2)=\delta(A_1,A_2)\big(1-(u/c)^2\big)<\delta(A_1,A_2).
$
\end{theorem}

Likewise in the Galilean geometry, we can consider  \eqref{poenkareove} as the affine mapping of $\mathcal A^2$.
Let us fix the inertial frame $\mathbf I$, i.e., the bijection $\varphi_\mathbf I$ that identifies $\mathcal A^2$
and $\mathbb R^2$. Then $P_c(u,\mathbf b)$ defines the \emph{pseudo-Euclidean isometry} -- an affine transformation of $\mathbb R^2$ that preserves the
pseudo-Euclidean quadratic form \eqref{quadraticForm}. 

\begin{figure}[ht]
{\centering
{\includegraphics[width=10.5cm]{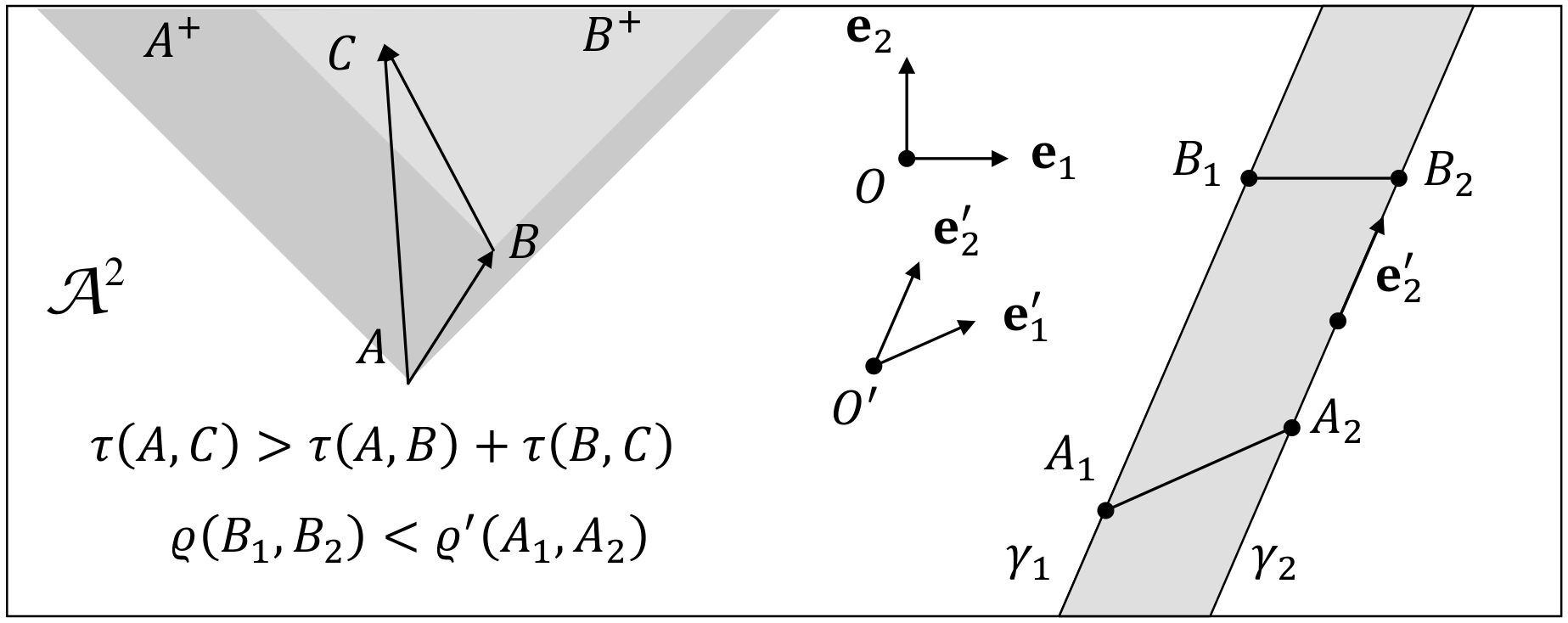}}
\caption{Twin paradox and contraction of length in the 2-dimensional pseudo-Euclidean world $\mathcal A^2$.
Here, the direction of time-like lines $\gamma_1$ and $\gamma_2$ is given by $\mathbf e_2'$.}
\label{diletacija-kontrakcija}}
\end{figure}

The most convenient and common is to use the natural system of units, where $c=1$.
Then we omit the subscript $c$ in the above notation:
$\mathrm L(u)=\mathrm L_1(u)=\tilde{\mathrm L}_1(u)$, $P(u,\mathbf b)=P_1(u,\mathbf b)=\tilde P_1(u,\mathbf b)$.

Note that the reflections in the space $\mathbf S$ and time $\mathbf T$,
\begin{align}
\label{Smap}&\begin{pmatrix}
x'\\
t'
\end{pmatrix}
=
\mathbf S
\begin{pmatrix}
x\\
t
\end{pmatrix}
=
\begin{pmatrix}
-1 & 0 \\
0 & 1
\end{pmatrix}
\begin{pmatrix}
x\\
t
\end{pmatrix}
=
\begin{pmatrix}
-x\\
t
\end{pmatrix}
, \\
\label{Tmap}&\begin{pmatrix}
x'\\
t'
\end{pmatrix}
=
\mathbf T
\begin{pmatrix}
x\\
t
\end{pmatrix}
=
\begin{pmatrix}
1 & 0 \\
0 & -1
\end{pmatrix}
\begin{pmatrix}
x\\
t
\end{pmatrix}
=
\begin{pmatrix}
x\\
-t
\end{pmatrix}
,
\end{align}
are not Lorentz transformations but they preserve the quadratic form \eqref{quadraticForm}.
The full group of the pseudo-Euclidean isometries $P(1,1)$ will be described in the next section.

\subsection{4-dimensional affine world.}
Let $\mathcal A^4$ be the affine world, $\mathbb V$ be the
associated vector space endowed with a symmetric, bilinear,
nondegenerate  form $\langle \cdot,\cdot\rangle\colon \mathbb
V\times \mathbb V \to \mathbb R$ of signature (3,1) (equivalently, one can consider the signature $(1,3)$).
This means that there exist a  basis $\mathbf e_1,\mathbf
e_2,\mathbf e_3,\mathbf e_4$ of $\mathbb V$, such that
\[
\langle \mathbf e_i,\mathbf e_j\rangle=\delta_{ij}, \quad i=1,2,3, \quad j=1,2,3,4, \quad \langle \mathbf e_4,\mathbf e_4\rangle=-1.
\]
All definitions are as in the 2-dimensional affine
world with $\mathbf e_2$ replaced by $\mathbf e_4$ (light-like cones, positive time-like vectors, spaces of simultaneous events
with direction $\mathbf e_4^\perp$, the Euclidean structure on $\mathbf e_4^\perp$, etc.).

\section{The full groups of isometries in 2-dimensional Galilean and pseudo-Euclidean worlds.}\label{pune}

Let us consider an affine world with a chosen reference frame $\mathbf I=[O,\mathbf e_1,\mathbf e_2]$ and consider the
associated coordinates $(x,t)$.
In the Galilean world we assume that the basis $\mathbf e_1,\mathbf
e_2$ is a Galilean basis, while in the pseudo-Euclidean case we
assume that the basis is pseudo-orthonormal. We also assume that
the speed of light is normalized: $c=1$.

\subsection{Galilean world.}

Recall that the Galilean structure in the plane $\mathbb R^2$ is determined by two functions $l$ and $\rho$:
\begin{align*}
& l(A(x_1,t_1)),B(x_2,t_2))=t_2-t_1 ,\\
& l(A,B)=0  \quad \Rightarrow \quad \rho(A(x_1,b_1),B(x_2,t_2))=\vert x_2-x_1\vert
\end{align*}
and that the Galilean transformations $\mathbf x'=G(u,\mathbf b)(\mathbf x)=\mathrm A(u)\mathbf x+\mathbf b$, $\mathrm A(u)\in SG_0(2)$ preserve both $l$ and $\rho$.

We have a natural question: Is $SG(2)$ the maximal group that preserve the Galilean structure?
In the definition of the inertial reference frames,
we imposed the condition that the orientation of $\mathcal A^2$ is fixed. This is the reason that the space reflection \eqref{Smap}
is not an element of $SG_0(2)$, although it preserves the Galilean structure.
On the other side, the reflection of the time coordinate \eqref{Tmap} does not preserve the function $l$.

Let $\mathbf SSG_0(2)=\{\mathbf S \circ \mathrm A(u)\vert\,u\in \mathbb R\}$ and $\mathbf SSG(2)=\{\mathbf S \circ G(u,\mathbf b)\vert\,u\in \mathbb R,\mathbf b\in\mathbb R^2\}$.
It is clear that $SG_0(2)\cap \mathbf S SG_0(2)=\emptyset$, i.e., $SG(2)\cap \mathbf S SG(2)=\emptyset$.

\begin{theorem}\label{pomocni}
Prove the following identities
\begin{align*}
& \mathrm A(u)\circ \mathbf S=\mathbf S\circ \mathrm A(-u),  \\
& \mathbf S\circ \mathrm A(u_1)\circ \mathbf S\circ \mathrm A(u_2)=\mathrm A(u_2-u_1),\\
& \mathbf S \circ G(u_1,\mathbf b_1)\circ \mathbf S\circ G(u_2,\mathbf b_2)=G(u_2-u_1, \mathrm A(-u)\mathbf b_2+\mathbf S\mathbf b_1).
\end{align*}
\end{theorem}

Whence, $G_0(2)=SG_0(2)\cup\mathbf S SG_0(2)$ and $G(2)=SG(2)\cup\mathbf S SG(2)$ are groups of linear and affine transformations.

\begin{prop}
$G(2)=SG(2)\cup\mathbf S SG(2)$ is the full group of affine transformations that preserve
the Galilean structure.
\end{prop}

\begin{proof}
The proof is similar to the proof of Proposition \ref{stav1}.
Let $S\colon \mathbb R^2\to\mathbb R^2$, $\mathbf x'=\mathrm A\mathbf x+\mathbf b$ be an affine transformation that preserves the Galilean structure.
Then the composition of $S$ and the translation $\mathrm R=T_{-\mathbf b}\circ S$
fix the origin, that it is a linear mapping of $\mathbb R^2$ that preserves the Galilean structure.

The space-like vector $\mathbf e_1$ have to be mapped by $\mathrm R$ to a space like vector
of the same magnitude.
Thus, either $\mathrm R\mathbf e_1=\mathbf e_1$, or $\mathrm R\mathbf e_1=-\mathbf e_1$.
On the other hand, from $l(R\mathbf e_2)=l(\mathbf e_2)=1$, we get that
$\mathrm R\mathbf e_2=\mathbf e_2+v\mathbf e_1$, for some $v\in\mathbb R$.
This proves that $\mathrm R$ belongs either in $SG_0(2)$ or in $\mathbf S SG_0(2)$, i.e., $S=T_{\mathbf b}\circ \mathrm R$ is an element of $G(2)$.
\end{proof}

\subsection{Pseudo-Euclidean isometries.}

Firstly, we consider linear transformations of the associated vector space
that preserve the pseudo-Euclidean scalar product
\begin{equation}\label{product}
\langle \mathbf u,\mathbf v\rangle=u_1v_1-u_2v_2, \quad \mathbf u=u_1\mathbf e_1+u_2\mathbf e_2, \quad \mathbf v=v_1\mathbf e_1+v_2\mathbf e_2.
\end{equation}

By $\mathbf E$ we denote the identity mapping (i.e., the identity matrix).
Since $\mathbf S\circ\mathbf T=\mathbf T\circ\mathbf S=-\mathbf E$, $\mathbf S^2=\mathbf T^2=\mathbf E$, the transformations
$\{\mathbf E,\mathbf S,\mathbf T,\mathbf S\circ\mathbf T\}$ form a group isomorphic to the Klein group $\mathbb Z_2\times\mathbb Z_2$.

Let us set
\begin{align*}
&\mathbf S SO^+(1,1)=\big\{\mathbf S\circ \mathrm L(u)=
\kappa(u)\begin{pmatrix}
-1 & -u  \\
{u} & 1 \end{pmatrix} \big\vert\, u\in (-1,1)\big\}, \\
&\mathbf T SO^+(1,1)=\big\{\mathbf S\circ \mathrm L(u)=
\kappa(u)\begin{pmatrix}
1 & u  \\
-{u} & -1 \end{pmatrix} \big\vert\, u\in (-1,1)\big\},\\
&\mathbf S\mathbf T SO^+(1,1)=\big\{-\mathbf E\circ \mathrm L(u)=
\kappa(u)\begin{pmatrix}
-1 & -u  \\
-{u} & -1 \end{pmatrix} \big\vert\, u\in (-1,1)\big\},
\end{align*}
where $\kappa(u)=1/\sqrt{1-{u^2}}$. It is clear that
$SO^+(1,1)$, $\mathbf S SO^+(1,1)$, $\mathbf T SO^+(1,1)$, $\mathbf S\mathbf T SO^+(1,1)$
do not intersect between themselves.

\begin{prop} \label{ortogonalnaGrupa}
The group of all linear transformations that preserve
the pseudo-Euclidean scalar product \eqref{product} is
\[
O(1,1)=SO^+(1,1)\cup\mathbf S SO^+(1,1)\cup \mathbf T SO^+(1,1)\cup\mathbf S\mathbf T SO^+(1,1).
\]
\end{prop}

\begin{proof}
Consider the linear mapping $\mathbf v'=\mathrm A\mathbf v$,
\[
\begin{pmatrix}
v_1'\\
v_2'
\end{pmatrix}=
\begin{pmatrix}
a & b \\
c & d
\end{pmatrix}
\begin{pmatrix}
v_1\\
v_2
\end{pmatrix}.
\]

It preserves the product \eqref{product} if and only if $\langle \mathrm A \mathbf e_i,\mathrm A\mathbf e_j\rangle =
\langle \mathbf e_i,\mathbf e_j\rangle$ ($i,j=1,2$), that is, if
\begin{equation*}\label{abcd}
a^2-c^2=1, \qquad b^2-d^2=-1, \qquad ab-cd=0.
\end{equation*}

From the first and the second equations we get
\[
(a,c)=(\pm \kappa(u),\pm u\kappa(u)), \quad
(b,d)=(\pm v\kappa(v),\pm \kappa(v)), \quad u,v\in(-1,1),
\]
while from the third equation we get $u=\pm v$.

Finally, the composition of two linear transformations that preserve the scalar product \eqref{product}, preserves
\eqref{product} as well. Thus, $O(1,1)$ is a group. The statement is proved.
\end{proof}

\begin{rem}
The columns of the matrixes $\mathrm A\in O(1,1)$ define all possible pseudo-orthonormal bases $\mathbf v,\mathbf w$ of $\mathbb R^2$:
$\mathbf v=\mathrm A\mathbf e_1=a\mathbf e_1+c\mathbf e_2$, $\mathbf w=\mathrm A\mathbf e_2=b\mathbf e_1+d\mathbf e_2$ (compare with Remark \ref{kolone}).
\end{rem}

The group $O(1,1)$ is called the \emph{orthogonal group of signature} $(1,1)$. The subgroup
\[
SO(1,1)=
SO^+(1,1)\cup\mathbf S\mathbf T SO^+(1,1)=O(1,1)\cap SL(2)
\]
is called the \emph{special orthogonal group of signature} $(1,1)$.

Similarly,
the full group $P(1,1)$ of affine transformations that preserve the pseudo-Euclidean quadratic form
\begin{equation}\label{quadraticForm*}
\delta(A(x_1,t_1),B(x_2,t_2))=(x_2-x_1)^2-(t_2-t_1)^2.
\end{equation}
is given by the transformations $\mathbf x'=\mathrm A\mathbf x+\mathbf b$, $\mathrm A\in O(1,1)$, $\mathbf b\in\mathbb R^2$:
\[
P(1,1)=SP^+(1,1)\cup \mathbf S SP^+(1,1) \cup\mathbf T SP^+(1,1)+\mathbf S\mathbf T SP^+(1,1).
\]

The groups $SG^+_0(2)$ and $SO^+(1,1)$ are {connected components of the identity matrix} $\mathbf E$
in $G_0(2)$ and  $O(1,1)$,  respectively.
The sets $\mathbf S SG^+_0(2)$ and $\mathbf S SO^+_0(1,1)$ are connected components of $\mathbf S$
within $G_0(2)$ and $O(1,1)$, while $\mathbf T SO^+_0(1,1)$ and $\mathbf S\mathbf T SO^+_0(1,1)$
are connected components of $\mathbf T$ and $-\mathbf E$ in the group $O(1,1)$.

In the next section we will visualise the special linear group $SL(2)$ and
its subgroups $SO(2)$, $SG_0(2)$, and $SO(1,1)$ at an elementary level.

\section{Visualisation of groups}\label{vizualizacija}

\subsection{The Iwasawa decomposition of $SL(2)$.}
The group $SL(2)$ can be seen as a set of bases $\mathbf v, \mathbf w$
of $\mathbb R^2$, such that the oriented area of the parallelogram spanned by $\mathbf v,\mathbf w$ is equal to 1 (see Remark  \ref{kolone}).
For a given $\mathbf v=\mathrm A\mathbf e_1$ and $\mathbf w=\mathrm A\mathbf e_2$, we define parameters $\alpha,\sigma,u$ in
\eqref{ID} as follows (see  Fig. \ref{sl(2)}): $\sigma=1/\vert \mathbf v\vert$,
$\mathbf e_1=\sigma\mathrm R(-\alpha)\mathbf v$. Further, $u$ is defined by the condition that the oriented parallelogram
$\Pi(\sigma^{-1}\mathbf e_1,\sigma \mathrm A(u)\mathbf e_2)$ is congruent to the oriented parallelogram $\Pi(\mathbf v,\mathbf w)$:
\footnote{The above consideration proves Problem \ref{probIWD}. What is missing here is the prof that the vector
$\mathrm R(-\alpha)\mathbf w$ belongs to the line $\gamma=\{\sigma \mathrm A(u)\mathbf e_2\,\vert u\in\mathbb R\}$. For a general Iwasawa decomposition
of semisimple Lie groups, e.g., see \cite{Kn}.}
\[
\mathrm R(-\alpha)\big(\Pi(\mathbf v,\mathbf w)\big)=\Pi(\sigma^{-1}\mathbf e_1,\sigma \mathrm A(u)\mathbf e_2).
\]

\begin{figure}[ht]
{\centering
{\includegraphics[width=10cm]{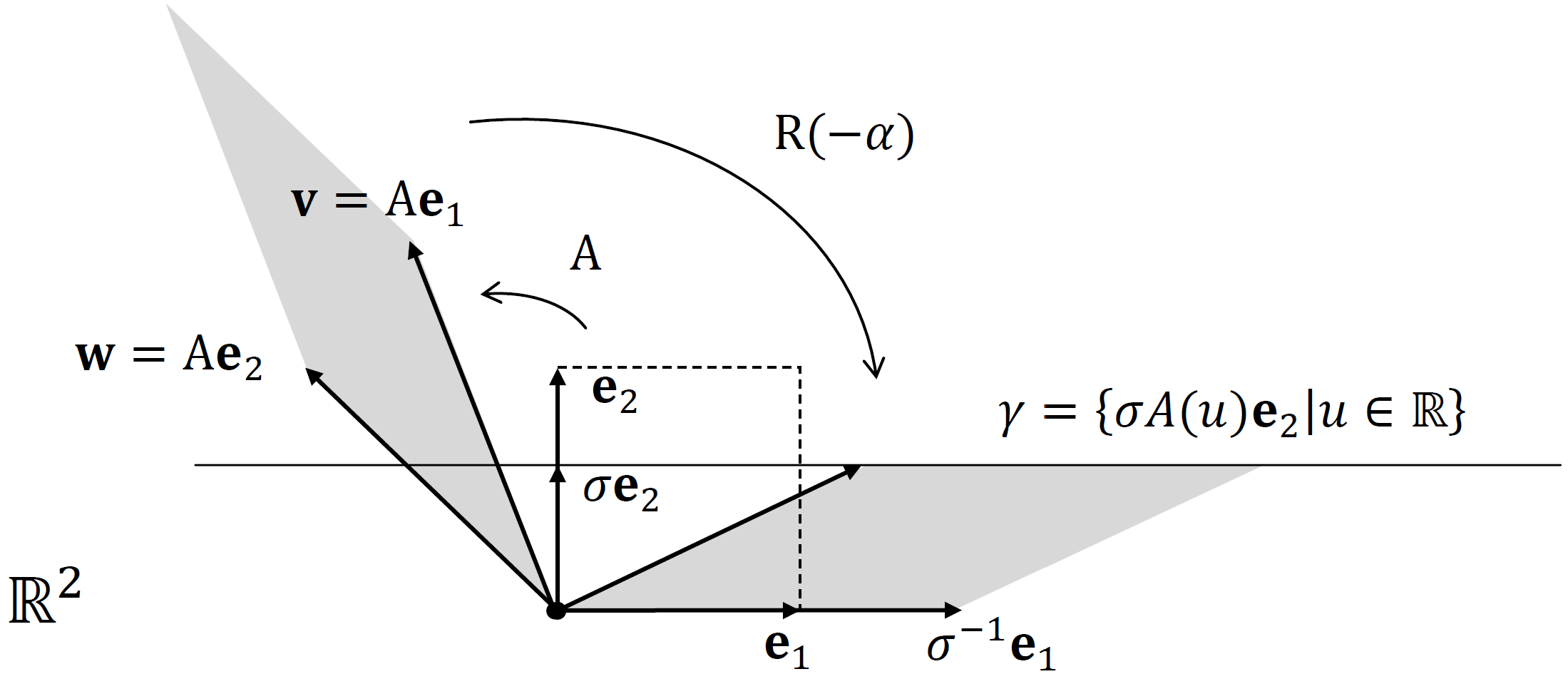}}
\caption{The Iwasawa decomposition of $SL(2)$.}\label{sl(2)}}
\end{figure}

The decomposition \eqref{ID} defines a bijection (moreover, $\phi$ is in example of a \emph{diffeomorphism}):
\[
\phi\colon S^1\times \mathbb R^+\times \mathbb R\rightarrow SL(2), \quad
\phi(\alpha,\sigma,u)=\mathrm R(\alpha)\circ \mathrm A(u)\circ \mathrm D(\sigma)
\]
($\alpha \in [0,2\pi]$,  $\sigma>0$, $u\in\mathbb R$), where we identified the points with $\alpha=0$ and $\alpha=2\pi$.

Next, we can use $(\alpha,\sigma,u)$ as cylindrical coordinates in $\mathbb R^3$:
\begin{align*}
& \psi\colon S^1\times \mathbb R^+\times \mathbb R\rightarrow \mathbb R^3, \qquad (x,y,z)=(\sigma\cos\alpha,\sigma\sin\alpha,u).
\end{align*}

In such a way, we obtain a bijection (a diffeomorphism):
\[
\Gamma=\psi\circ\phi^{-1}\colon SL(2)\rightarrow \mathbb R^3_*,
\]
where $\mathbb R^3_*$ is a 3-dimensional space without the $z$-axis $\{x=y=0\}$.
For example, the unit matrix $\mathbf E$ and the matrix $-\mathbf E$ are represented by the points $(1,0,0)$ and $(-1,0,0)$, respectively.

The images of the special orthogonal group $SO(2)$ and the group of Galilean rotations are given by
\begin{align*}
&  \Gamma(SO(2))=\{(\cos\alpha,\sin\alpha,0)\,\vert\, \alpha\in [0,2\pi)\},\\
&  \Gamma(SG_0(2))=\{(1,0,u)\,\vert\, u\in\mathbb R\}.
\end{align*}

The image of a special orthogonal group $SO(2)$ is a unit circle centered at the origin
and we have a visualisation of the \emph{fundamental group} $\pi_1(SL(2))=\mathbb Z$ of $SL(2)$ that is generated by $SO(2)$
(by smooth deformations within $SL(2)$, the circle $SO(2)$ can not shrink to a point), see Fig. \ref{iwd}.

\begin{figure}[ht]
{\centering
{\includegraphics[width=6.5cm]{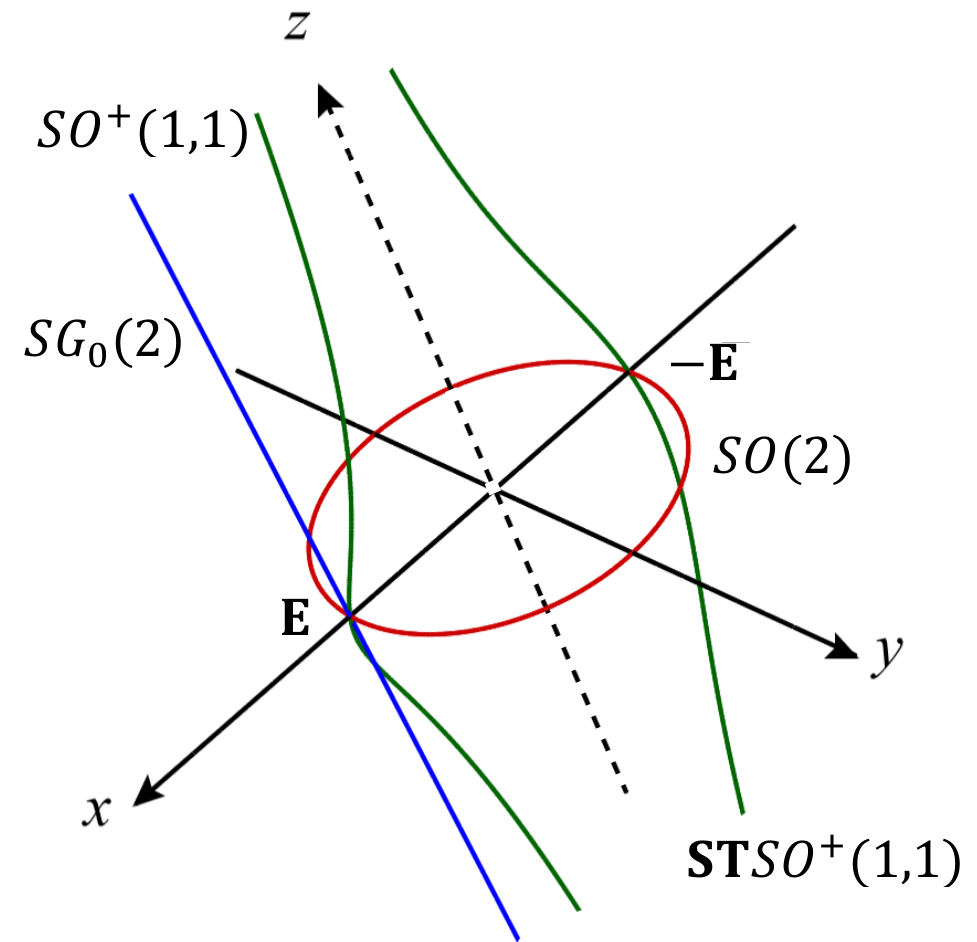}}
\caption{The realisation of $SL(2)$ within $\mathbb R^3$ without the $z$-axis. $SO(2)$ is a unit circle in the $xy$-plane with the center at the origin.
$SO_0(2)$ is a line orthogonal to the $xy$-plane through $\mathbf E$. Further,
$SO^+(1,1)$ intersects $SO(2)$ and $SG_0(2)$ at $\mathbf E$, while $\mathbf S\mathbf T SO^+(1,1)$ intersects $SO(2)$ at $-\mathbf E$.}\label{iwd}}
\end{figure}

\begin{theorem} Prove that the Lorentz transformations can be represented as
\[
\mathrm L(u)=
\begin{pmatrix}
\frac{1}{\sqrt{1+u^2}} & \frac{-u}{\sqrt{1+u^2}}  \\
\frac{u}{\sqrt{1+u^2}} & \frac{1}{\sqrt{1+u^2}}
\end{pmatrix}
\circ
\begin{pmatrix}
1 & \frac{2u}{1-u^2}  \\
0 & 1 \end{pmatrix}\circ
\begin{pmatrix}
\frac{\sqrt{1+u^2}}{\sqrt{1-u^2}} & 0  \\
0
 & \frac{\sqrt{1-u^2}}{\sqrt{1+u^2}} \end{pmatrix},
\]
where $u\in(-1,1)$.
Equivalently, by defining
$\alpha=\arcsin({u}/{\sqrt{1+u^2}})$,
the decomposition takes the form
\[
\mathrm L(u)=
\begin{pmatrix}
\cos\alpha & -\sin\alpha  \\
\sin\alpha & \cos\alpha
\end{pmatrix}
\circ
\begin{pmatrix}
1 & \tan2\alpha  \\
0 & 1 \end{pmatrix}\circ
\begin{pmatrix}
\frac{1}{\sqrt{\cos2\alpha}} & 0  \\
0
 & {\sqrt{\cos2\alpha}} \end{pmatrix}.
\]
\end{theorem}

Therefore, the image of the Lorentz group $SO^+(1,1)$ is  (see Fig. \ref{iwd}):
\[
\Gamma(SO^+(1,1))=\big\{\big({\cos\alpha}{\sqrt{\cos2\alpha}},{\sin\alpha}{\sqrt{\cos2\alpha}},\tan2\alpha\big)\,\vert\,
\alpha\in\big(-\frac{\pi}{4},\frac{\pi}{4}\big)\big\}.
\]

Further, since
$
-\mathbf E \circ \mathrm R(\alpha)\circ \mathrm A(u)\circ \mathrm D(\sigma)=\mathrm R(\pi+\alpha)\circ \mathrm A(u)\circ \mathrm D(\sigma),
$
the image of the component of $SO(1,1)$ that contains the matrix $-\mathbf E$ is given by:
\begin{align*}
\Gamma(\mathbf S\mathbf T SO^+(1,1))=\big\{\big({\cos(\alpha+\pi)}{\sqrt{\cos2\alpha}},{\sin(\alpha+\pi)}{\sqrt{\cos2\alpha}},\tan2\alpha\big)\,\vert\,
\alpha\in\big(-\frac{\pi}{4},\frac{\pi}{4}\big)\big\}.
\end{align*}

\subsection{$SL(2)$ as a quadric in $\mathbb R^4$.}
The group $SL(2)$ is an algebraic subset of $\mathbb R^4$ with coordinates $(a,b,c,d)$ given by the quadratic equation $ad-bc=1$ (see \eqref{sl2}).
We can visualise it by considering the intersections with different 3-dimensional affine subspaces $\Sigma\subset \mathbb R^4$
and projecting $SL(2)\cap \Sigma$ to the 3-dimensional space
$\mathbb R^3$ with coordinates $(a,d,b)$. We present the following quadratic surfaces:

\begin{itemize}
\item[(i)]  $\Sigma=\{ c=0\}$, $SL(2)\cap \Sigma=\{ad-bc=1, \, c=0\}\rightarrow \mathcal H_1=\{ad=1\}\subset\mathbb R^3 $
\item[(ii)]    $\Sigma=\{c=1\}$, $SL(2)\cap \Sigma=\{ad-bc=1, \, c=1\}\rightarrow \mathcal H_2=\{ad-b=1\}\subset\mathbb R^3 $
\item[(iii)]  $\Sigma=\{c=b\}$, $SL(2)\cap \Sigma=\{ad-bc=1, \, c=b\}\rightarrow \mathcal H_3=\{ad-b^2=1\}\subset\mathbb R^3 $
\item[(iv)] $\Sigma=\{ c=-b\}$, $SL(2)\cap \Sigma=\{ ad-bc=1, \, c=-b\}\rightarrow \mathcal H_4=\{ ad+b^2=1\}\subset\mathbb R^3 $
\end{itemize}

\begin{figure}[ht]
{\centering
{\includegraphics[width=11cm]{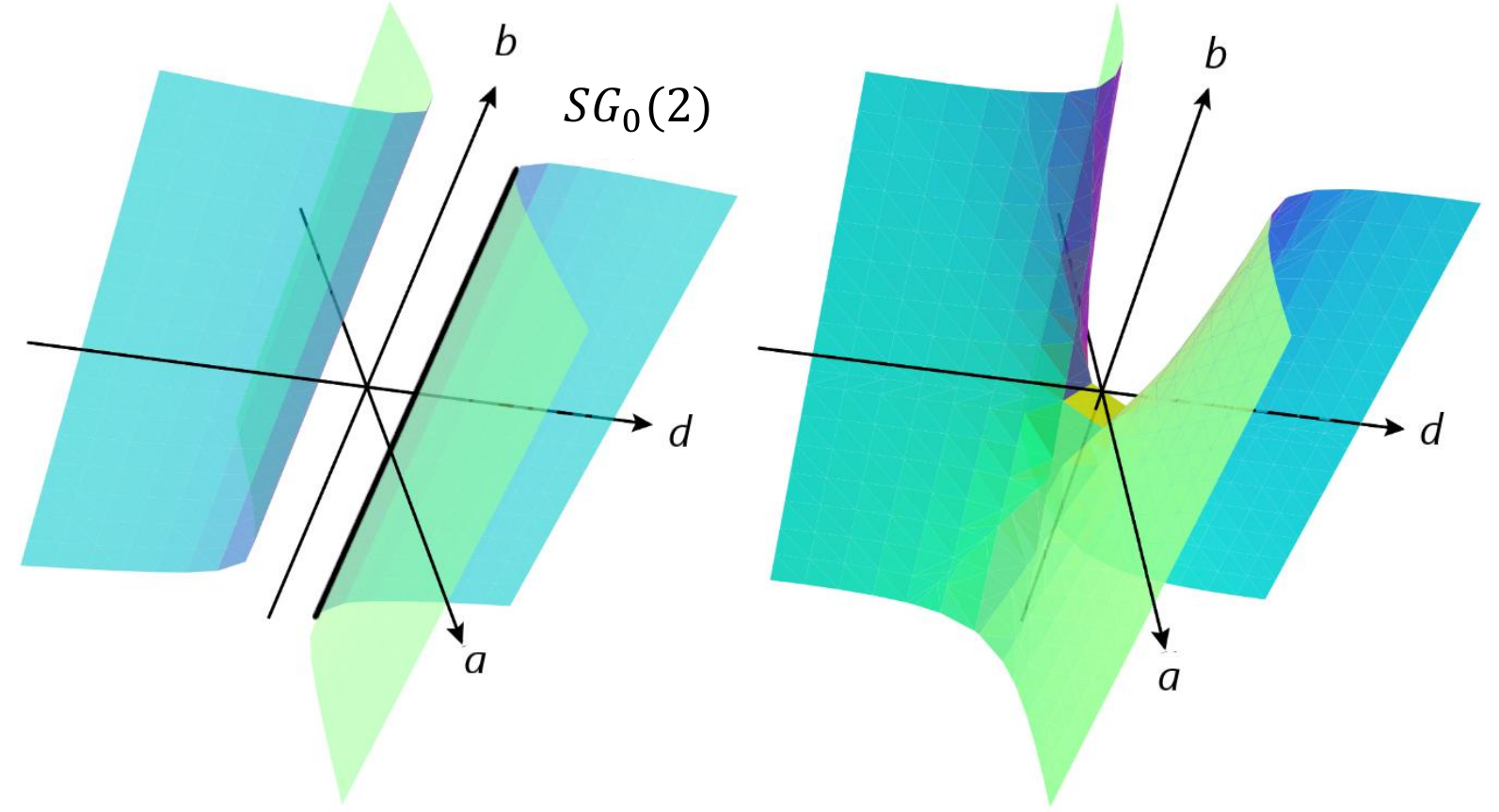}}
\caption{The cases (i) and (ii): a hyperbolic cylinder $\mathcal H_1$ and a hyperbolic paraboloid $\mathcal H_2$.}}
\end{figure}

\begin{figure}[ht]
{\centering
{\includegraphics[width=10.5cm]{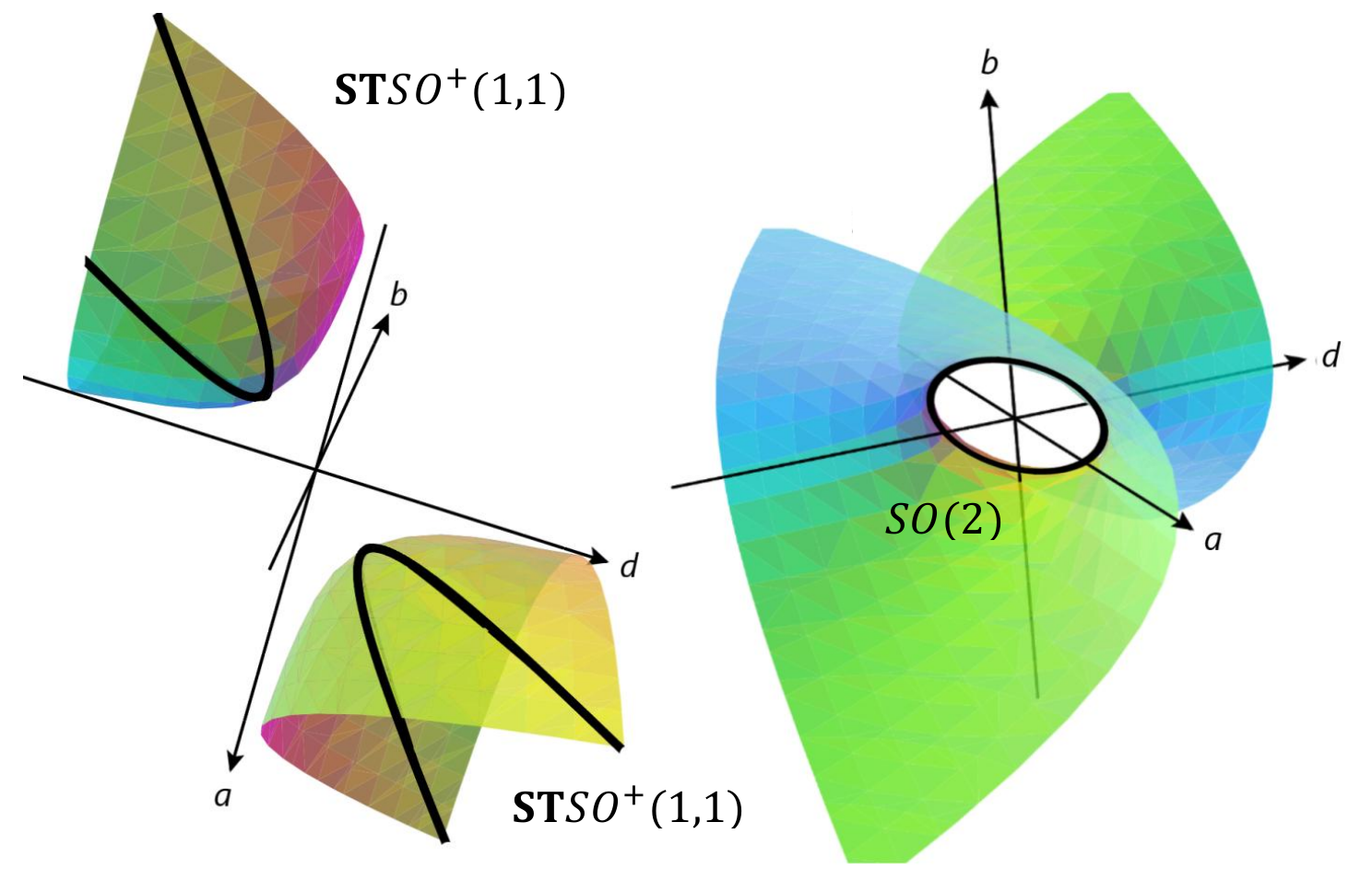}}
\caption{The cases (iii) and (iv): a hyperboloid of two sheets $\mathcal H_3$ and  a hyperboloid of one sheet $\mathcal H_4$.}}
\end{figure}

Within the hyperbolic cylinder $\mathcal H_1$ (the case (i)) we can visualise the group of Galilean rotations $SG_0(2)$
as a line
\[
SG_0(2)\subset SL(2)\cap \Sigma \rightarrow  \{(1,1,u)\,\vert\,u\in\mathbb R\}\subset \mathcal H_1.
\]
On the other hand, we can visualise $SO(1,1)$ as a intersection of the plane $\Lambda=\{a=d\}\subset \mathbb R^3$ and a hyperboloid of two sheets $\mathcal H_3$.
The component of the intersections with $a=d>0$ represents the group of Lorentz transformations $SO^+(1,1)$, while the component
with $a=d<0$ represents the subset $\mathbf S\mathbf T SO^+(1,1)$. This follows from the following problem:

\begin{theorem}
Prove that special orthogonal group  $SO(1,1)$ is equal to $SL(2)\cap K_1$,
where $K_1$ is the group of $2\times 2$ symmetric real matrixes with equal elements at the diagonal and determinant greater then $0$ (see Problem \ref{Kc}).
\end{theorem}

In the case (iv) we can see the special orthogonal group $SO(2)$,
\[
SO(2)\subset SL(2)\cap \Sigma \rightarrow \{(\cos\alpha,\cos\alpha,\sin\alpha)\,\vert\,\alpha\in[0,2\pi)\}\subset\mathcal H_4,
\]
and we obtain another  visualisation of the {fundamental group} $\pi_1(SL(2))=\mathbb Z$ generated by $SO(2)$.

Finally, let us note that many interesting problems of elementary geometry in the Galilean and the pseudo-Euclidean plane
can be found in \cite{Y}.

\subsection*{Acknowledgments}
The author is very grateful to Vladimir Dragovi\' c for useful discussions.
The research was supported by the Project no. 7744592 MEGIC ”Integrability
and Extremal Problems in Mechanics, Geometry and Combinatorics” of the Science Fund
of Serbia.

\end{document}